\documentclass[reprint,superscriptaddress,bibnotes,amsmath,amssymb,aps,floatfix,longbibliography]{revtex4-1}

\usepackage{qcircuit}
\usepackage{amsmath,bm}
\usepackage{graphicx}
\usepackage{amsmath,amssymb,amsthm,mathrsfs,amsfonts,dsfont}
\usepackage{braket}
\usepackage{bm}
\usepackage{enumerate}
\usepackage{color}
\usepackage{graphicx}
\usepackage{algorithm}
\usepackage{algorithmic}
\usepackage{braket}
\usepackage{comment}
\usepackage{appendix}
\usepackage{here}
\usepackage{tabularx}
\usepackage{ulem}
\newtheorem{theorem}{Theorem}

\newtheorem{lemma}{Lemma}
\newtheorem{proposition}{Proposition}
\newtheorem{assumption}{Assumption}

\DeclareRobustCommand{\YMdel}{\bgroup\markoverwith{\textcolor[rgb]{0.1, 0.5, 0.1}{\rule[.5ex]{2pt}{0.4pt}}}\ULon}
\usepackage[colorlinks,linkcolor=blue,citecolor=blue]{hyperref}

\usepackage[justification=RaggedRight]{caption}
\usepackage[subrefformat=parens]{subcaption}
\newcommand{\tr}{\mathrm{Tr}}

\begin{document}


\title{Localized Virtual Purification}


\author{Hideaki Hakoshima}
\email{hakoshima.hideaki.es@osaka-u.ac.jp}
\affiliation{Graduate School of Engineering Science, Osaka University, 1-3 Machikaneyama, Toyonaka, Osaka 560-8531, Japan}
\affiliation{Center for Quantum Information and Quantum Biology, Osaka University, 1-2 Machikaneyama, Toyonaka, Osaka 560-0043, Japan}

\author{Suguru Endo}
\affiliation{NTT Computer and Data Science Laboratories, NTT Corporation, Musashino 180-8585, Japan}
\affiliation{JST, PRESTO, 4-1-8 Honcho, Kawaguchi, Saitama, 332-0012, Japan}

\author{Kaoru Yamamoto}
\affiliation{NTT Computer and Data Science Laboratories, NTT Corporation, Musashino 180-8585, Japan}

\author{Yuichiro Matsuzaki}
\affiliation{Department of Electrical, Electronic, and Communication Engineering, Faculty of Science and Engineering, Chuo University, 1-13-27 Kasuga, Bunkyo-ku, Tokyo 112-8551, Japan.}

\author{Nobuyuki Yoshioka}
\email{nyoshioka@ap.t.u-tokyo.ac.jp}
\affiliation{Department of Applied Physics, University of Tokyo, 7-3-1 Hongo, Bunkyo-ku, Tokyo 113-8656, Japan}
\affiliation{Theoretical Quantum Physics Laboratory, RIKEN Cluster for Pioneering Research (CPR), Wako-shi, Saitama 351-0198, Japan}
\affiliation{JST, PRESTO, 4-1-8 Honcho, Kawaguchi, Saitama, 332-0012, Japan}



\begin{abstract}
Analog and digital quantum simulators
can efficiently simulate quantum many-body systems that appear in natural phenomena.
However, experimental limitations of near-term devices still make it challenging to perform the entire process of quantum simulation.
The purification-based quantum simulation methods can alleviate the limitations in experiments such as the cooling temperature and noise from the environment, while this method has the drawback that it requires global 
entangled measurement
with a prohibitively large number of measurements that scales exponentially with the system size.
In this Letter, we propose that we can overcome these problems by restricting the entangled measurements to the vicinity of the local observables to be measured, when the locality of the system can be exploited. 
We provide theoretical guarantees  that the global purification operation can be replaced with local operations under some conditions, in particular for the task of cooling and error mitigation.
We furthermore give a numerical verification that the localized purification is valid  even when conditions are not satisfied.
Our method bridges the fundamental concept of locality with quantum simulators, and therefore expected to open a path to unexplored quantum many-body phenomena.

\end{abstract}

\maketitle

\textit{Introduction.---} 
%
Simulating quantum many-body systems is a fundamental issue for 
quantum information science~\cite{feynman1982simulating}, 
since it potentially has a significant impact on various fields~\cite{RevModPhys.86.153} including condensed matter physics~\cite{PhysRevLett.79.2586,doi:10.1126/science.1207239,babbush2023exponential}, statistical physics~\cite{Trotzky_2012,doi:10.1126/science.aaf6725,doi:10.1126/science.aaf8834,rosenberg2023dynamics}, quantum chemistry~\cite{doi:10.1126/science.1113479,hastings2014improving,RevModPhys.92.015003}, and high-energy physics~\cite{doi:10.1126/science.1217069,PhysRevX.3.041018,PRXQuantum.3.020324}.
In particular,  simulation of thermal equilibrium states, ground states, and non-equilibrium dynamics for quantum many-body Hamiltonians has attracted attention as a valuable application, since it is believed to be an exponentially difficult task on a classical computer. 
This has motivated the recent progress in quantum simulations using cold atoms in an optical lattice~\cite{greiner2002quantum,Simon_2011,doi:10.1126/science.aal3837,Sch_fer_2020,doi:10.1126/science.abo6587}, nitrogen-vacancy centers in diamond~\cite{abobeih2022fault}, photonic devices~\cite{madsen2022quantum}, and superconducting qubits~\cite{PhysRevResearch.5.013183,kim2023evidence}.


While it remains a challenge to perform all quantum tasks in the current quantum devices,
it has been proposed that purification-based quantum simulation enables us to break the limitations in experiments.
The key idea is to enhance the purity of a quantum state in a virtual way by utilizing the classical post-processing, rather than directly realizing the purified quantum state.
More specifically, one computes the expectation value
$\braket{\bullet}_{\rm FVP}=\tr[\rho^{(n)}_{\rm FVP} \bullet ]$ corresponding to $\rho^{(n)}_{\rm FVP}=\rho^n/ \tr{[\rho^n]}$ (which is denoted as Fully Virtual Purification (FVP) throughout this paper) from an original quantum state $\rho$ using $n$ copies.
It has been pointed out that, such an operation is capable of (i) simulating the canonical Gibbs state of temperature $T/n$ using that of $T$~\cite{PhysRevX.9.031013}, and (ii) suppressing the effect of noise in the context of quantum error mitigation~\cite{PhysRevX.11.031057,PhysRevX.11.041036,https://doi.org/10.48550/arxiv.2210.10799,PhysRevLett.129.020502,PhysRevX.9.031013,https://doi.org/10.48550/arxiv.2210.10799,ohkura2023leveraging, Endo_2021,cai2022quantum,qin2022overview,PhysRevLett.119.180509,PhysRevX.8.031027,kandala2019error}. 
However, these methods require multiple entangled measurement gates
that act globally among multiple copies. This imposes a severe burden on the computation; non-local entangling gates among copies and an exponentially large number of measurements.
It is crucial to seek whether we can alleviate the overhead of purification-based methods in a way that the accuracy of the simulation is maintained. 

One promising direction is to utilize the geometrical locality of target models, which is present ubiquitously in condensed matter systems. In particular, the locality of interaction yields an upper bound on the velocity of information propagation; the Lieb-Robinson bound~\cite{lieb28finite,PhysRevB.69.104431,Nachtergaele_2006,hastings2010locality}. Recent works show that this powerful bound can be applied to yield various fundamental limits such as the finite correlation length of a gapped ground state~\cite{PhysRevB.69.104431,Nachtergaele_2006,hastings2010locality} and approximation of time-evolution unitary in the interaction picture~\cite{Haah_2021,PhysRevX.9.031006,PhysRevX.11.011020}.
Meanwhile, to our knowledge, there are very few frameworks of practical quantum algorithms other than Hamiltonian simulation~\cite{Haah_2021,PhysRevX.9.031006,PhysRevX.11.011020, PRXQuantum.3.040302} that incorporate the notion of geometrical locality. This implies that we are not fully harnessing the capacity of the quantum simulators for practical use. 

In this Letter, we fill in these gaps by  proposing the Localized Virtual Purification (LVP) as a virtual purification on local subsets of qubits, and present theoretical guarantees and conditions that the method overcomes the problems in FVP when the locality of the system can be exploited. 
While the proposal itself was mentioned in the original paper as a task for cooling~\cite{PhysRevX.9.031013}, our contribution is to clarify the conditions of the theoretical bounds and to show a practical advantage to an effect of noise among the entanglement measurement operations.
While the output from the LVP generally deviates from that of the FVP 
due to a non-negligible correlation between purified and unpurified regions,
we find that the deviation can be written as a generalized correlation function.
In particular, this reduces to the two-point correlation function in some cases including cooling and error mitigation. Therefore, if we further assume the clustering property, i.e., the exponential decay of two-point correlation, we can derive two bounds that assure the accuracy of the LVP for these tasks.
Finally, we verify our analytical results via numerical simulation, and also find that the LVP is capable of unifying the two tasks, namely the  simulation of low-temperature Gibbs states from  noisy high-temperature states.
\begin{figure}[t!]
    \centering
    \includegraphics[width=0.48\textwidth]{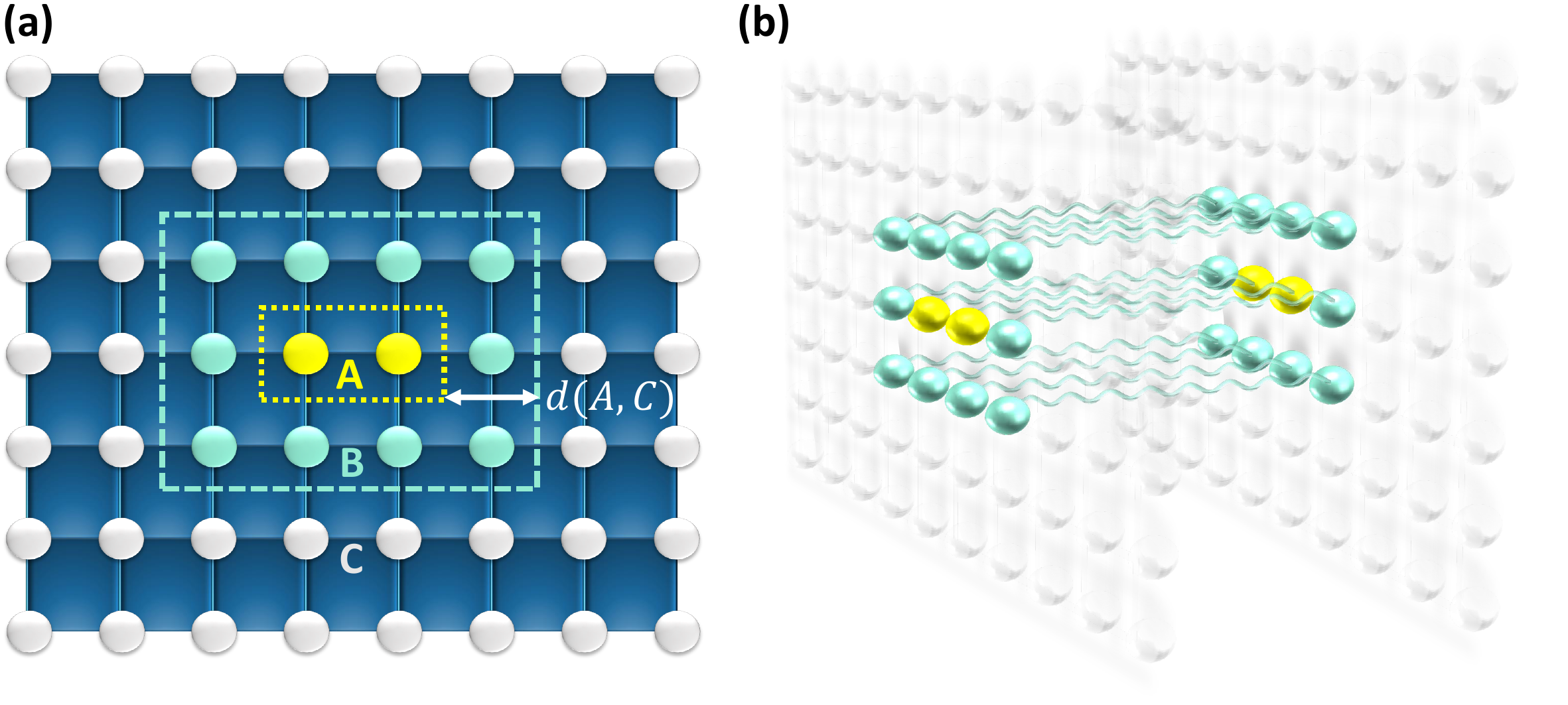}
    \caption{Schematic description of our proposal, the Localized Virtual Purification (LVP), assuming two-dimensional system with $n=2$ copies of quantum state. (a) Dividing the entire lattice into three regions $A$, $B$, and $C$ according to the local observable $o_A$ and the distance $d(A,C)=1$.
    The number of sites inside the regions $A$, $B$, and $C$ are $N_A$, $N_B$ and $N_C$, respectively, and the total number of sites is $N(=N_A+N_B+N_C)$. (b) 
    The  entangled measurement operation  applied to regions $A$ and $B$  to perform the LVP.
    }
    \label{Fig:fig1}
\end{figure}

\textit{Setup.---}
Let the Hamiltonian be defined on a $d$-dimensional hypercubic lattice with the periodic boundary condition which is represented as an undirected graph $G=(V,E)$, where the number of vertices $|V|$ is equivalent to the total number of sites (or qubits) $N$ and the edges $E$ denote the connectivity of each site. 
We assume that the interactions in the Hamiltonian are geometrically local:
$H=\sum_{\bm X} h_{\bm X}$ satisfying $\max_{v\in V}\sum_{{\bm X}:{\bm X}\ni v} \|h_{\bm X}\| \le g$, where ${\bm X}$ is a subset of $V$ and $g$ is a positive constant independent of $N$.
In the following, we will write each term of the Hamiltonian
as $h_{A_i}$ with a subset $A_{i}$ denoting one vertex $i$ on the support,
as long as there is no confusion. 
We assume a spin-$1/2$ system that directly corresponds to qubits, while we expect that our theory can be naturally extended to higher-spin, fermionic, and bosonic systems, since the underlying mechanism of cluster property is commonly present.
In fact, we numerically analyze two-dimensional fermionic systems in the Supplementary Materials (SM)~\cite{Supplement_localized_virtual_purification}.

\textit{Localized virtual purification.---}
We define the expectation values of LVP by
\begin{align}
    \braket{\mathcal{O}}^{(n)}_{\rm LVP}&=\sum_i \frac{\mathrm{Tr}_{A_i+B_i}\left[ \left(\rho^{(A_i+B_i)}_i\right)^n o_{A_i}\right]}{\mathrm{Tr}_{A_i+B_i}\left[ \left(\rho^{(A_i+B_i)}_i\right)^n \right]},
    \label{eq:partialvirtual}
\end{align}
where $\mathcal{O}$ is an observable that can be decomposed into a sum of local observables as $\mathcal{O}=\sum_i o_{A_i}$.
Here, $\rho^{(A_i+B_i)}_i=\tr_{C_i}[\rho]$ is the reduced density operator on regions $A_i$ and $B_i$.
To be more specific, we divide all sites into three regions (see Fig.~\ref{Fig:fig1} (a)): $A_i$ is a region of the support of a local observable $o_{A_i}$ and we perform the entangled measurement operations between $n$ copies in the $A_i$ and $B_i$ regions, while we do not perform any operations on the region $C_i$.
Without loss of generality, we may focus on an arbitrary single term $o_{A_i}$ and therefore we omit the index $i$ for simplicity.
Note that the expectation values in Eq.~(\ref{eq:partialvirtual}) can be evaluated via LVP by operating an entangled measurement circuit on the region $A\cup B$
as shown in Fig.~\ref{Fig:fig1} (b). This results in the measurement costs of $\mathrm{Tr}_{A+B}[(\rho^{(A+B)}) ^n]^{-2}$~\cite{Supplement_localized_virtual_purification},
where measurement costs denote the increase of the variance for estimating the expectation values. The increase implies the greater number of measurements required.

In general, the estimation by the LVP is not equivalent to that of the FVP. 
Meanwhile, we find that we can understand the performance of LVP qualitatively by rewriting the difference of the expectation values  $D^{(n)}(o_A)=\frac{\mathrm{Tr}_{A+B}[(\rho^{(A+B)})^n o_A]}{\mathrm{Tr}_{A+B}[(\rho^{(A+B)})^n]}-\frac{\mathrm{Tr}[\rho^n o_A]}{\mathrm{Tr}[\rho^n]}$ into the integration of the generalized correlation function as 
\begin{align}
    D^{(n)}(o_A)
    &= \int_0^1 d\lambda\int_0^1 d\tau ~\mathrm{Corr}_{\rho_\lambda}^\tau (X_n-Y_n,o_A),    \label{eq:noisytwopointcorrelationfunction}
\end{align}
where we have introduced a generalized correlation function
    $\mathrm{Corr}_\rho^\tau (O,O') =\mathrm{Tr}[\rho^\tau O \rho^{1-\tau}O'] - \mathrm{Tr}[\rho O]\mathrm{Tr}[\rho O']$. This is also referred to as the two-point quantum correlation function~\cite{PhysRevLett.117.130401,Chiara_2018} or canonical correlation for thermal Gibbs states~\cite{Kubo_1966,kubo2012statistical,PhysRevA.94.062316} in literature.
We have also defined $\rho_\lambda=e^{ H_\lambda}/\tr[e^{ H_\lambda}] $, 
    $H_\lambda = Y_n +\lambda(X_n-Y_n)$
, $Y_n =\log{\left(\rho^n \right)}$, and
$X_n =\log{\left[\left(\rho^{(A+B)}\right)^n \otimes (\sigma_C)^n\right]}$, where $\sigma_C$ is an arbitrary positive operator (see SM for derivation~\cite{Supplement_localized_virtual_purification}).
As we discuss in detail in the following, we find 
that the generalized correlation function reduces to the two-point correlation function for some cases, and thus it decays exponentially if the exponential clustering property holds.

\textit{Theoretical bounds for LVP.---} Now let us present our main results: the performance guarantees and their conditions for the LVP.
We first discuss one of the most important cases, namely the cooling of thermal equilibrium states
\begin{align}
    \rho_{\beta}(H) = \frac{e^{-\beta H}}{\tr[e^{-\beta H}]},
    \label{eq:gibbs}
\end{align}
where $\beta$ is an inverse temperature. 
In this case, we can rewrite Eq.~\eqref{eq:noisytwopointcorrelationfunction} as
    $D^{(n)}(o_A)
    = -n\beta\int_0^1 d\lambda\int_0^1 d\tau ~\mathrm{Corr}_{\rho_\lambda}^\tau (\Delta H,o_A)$,
because simple relations $X_n =-n\beta (H+\Delta H)$ and $
    Y_n =-n\beta H$ hold for $\Delta H$ that is introduced to describe the deviation of the Hamiltonian from a Gibbs state $\rho_{\beta}(H^{(A+B)})$ tracing out the region $C$.
Here, we assume that $\Delta H$ can be approximated by a local operator supported on the boundary of the region $A\cup B$ with exponentially small error, as shown  in SM~\cite{Supplement_localized_virtual_purification}.
Note that although previous results have proved the above assumption based on a generalized linked-cluster expansion in order to evaluate $H+\Delta H$~\cite{PhysRevLett.124.220601,anshu2021sample,KUWAHARA2020168278}, which is related to the Hamiltonian of mean force for strong coupling systems~\cite{10.1063/1.1749657,Chris_Jarzynski_2004,PhysRevLett.102.210401,10.1116/5.0073853}, a flaw in its proof was pointed out later~\cite{kuwahara2024qip}.
While the proof is to be fixed, we reasonably expect that the locality of $\Delta H$ is valid in high temperature regime, as proven in commuting Hamiltonians~\cite{bluhm2024strong}. 
By noting that the support of $o_A$ is separate from that of $\Delta H$ under this assumption,
we prove the following theorem:
\begin{theorem}\label{Theorem1}(informal)
Assuming the exponential clustering of the two-point correlation function~\cite{araki1969gibbs,1995JSP....80..223P,ueltschi2005cluster,PhysRevX.4.031019,Fr_hlich_2015}, 
the deviation $|D^{(n)}(o_A)|$ is exponentially small in terms of $d(A,C)$:
\begin{align}
    &|D^{(n)}(o_A)|=O(e^{-d(A,C)}).
    \label{Eq:exponentialclusteringGibbs}
\end{align}
\end{theorem}
The formal statement and its proof are shown in SM~\cite{Supplement_localized_virtual_purification}.
%
%
One of the most outstanding points of our LVP protocol applied to Gibbs states is the suppression of measurement costs: 
    $\tr_{A+B}\left[\left(\rho^{(A+B)}\right)^n\right]^{-2}\simeq\exp{[2n\beta(N_A+N_B)(f_{n\beta}-f_{\beta})]}$,
where $f_\beta = -\frac{1}{(N_A+N_B)\beta}\mathrm{Tr}[e^{-\beta H_{A+B}}]$ is a free energy density of the Hamiltonian $H^{(A+B)}$ at an inverse temperature $\beta$.
This implies that the measurement costs of LVP are exponentially small regarding $N_C$ compared with those of FVP~\cite{Supplement_localized_virtual_purification}. Our LVP protocol for Gibbs states induces an exponentially small bias from FVP regarding $d(A,C)$, while it exponentially reduces the measurement costs.
Whereas it seems that the above assumption regarding $\Delta H$ only hold for an extremely high temperature region,
 as we later confirm in numerical simulations, this is expected to hold even for a low temperature region
 ~\cite{one_myfoot}.

Next, we discuss another interesting application of the LVP---error mitigation. Here, we first provide a theoretical bound when $\rho$ is pure, and then  present results when the noise is present. 

When the target state is pure, the deviation $|D^{(n)}(o_A)|$ can be written as a simpler form of the two-point correlation function:
    $D^{(n)}_0(o_A)=\mathrm{Tr}_{\rm A+C}[(\rho^{(A+C)}-\rho^{(A)}\otimes \rho^{(C)})o_A \otimes(\rho^{(C)})^{n-1}/\mathrm{Tr}[(\rho^{(C)})^n]]$ \cite{Supplement_localized_virtual_purification}.
Here, $D^{(n)}_0(o_A)$ denotes the deviation $D^{(n)}(o_A)$ for a pure state $\rho_0$, and $\rho^{(A)}$ and $\rho^{(C)}$ denote the reduced density operator of the region $A$ and $C$.
Regarding the deviation $D^{(n)}(o_A)$ for pure states and noisy states under global depolarizing noise channel, we prove the following theorem.
\begin{theorem}\label{Theorem2}
If the ground state is unique with a finite energy gap between the first excited states, the deviation for a noiseless pure state $\rho_0$ can be rewritten as
\begin{align}
    |D^{(n)}_0(o_A)|\le c \|o_A\| \frac{\|(\rho^{(C)})^{n-1}\|}{\mathrm{Tr}[(\rho^{(C)})^n]} \exp{\left(-\frac{d(A,C)}{\xi}\right)},
    \label{Eq:exponentialclustering}
\end{align}
where $c$ and $\xi$ are constants independent of $N$.
In particular, if the system of interest is a one-dimensional system, the term $\|(\rho^{(C)})^{n-1}\|/\mathrm{Tr}[(\rho^{(C)})^n]$ can be bounded by a constant independent of $d(A,C)$ for any $n$, which leads to $|D^{(n)}_0(o_A)|\le c' \|o_A\| \exp{\left(-d(A,C)/\xi\right)}=O(e^{-d(A,C)})$, where $c'$ is a constant independent of $N$.
Furthermore, the deviation for the noisy ground state under the global depolarizing noise can be bounded as $|D^{(n=2)}(o_A)|
    \le |D^{(n=2)}_0(o_A)|+ \left|\delta_1\right| + \left|\delta_2\right|$, and $\delta_1=O(e^{-N})$ and $\delta_2 =O(e^{-(N_A+N_B)})$.
    The global depolarizing noise channel is defined by $\mathcal{D}[\rho]=(1-p)\rho+p \frac{\mathbb{I}}{2^N}$,
where $p$ ($0\le p\le 1$) is the error rate.
\end{theorem}
\begin{proof}
The nontrivial part of this theorem has been done by explicitly showing the expression of Eq.~(\ref{eq:noisytwopointcorrelationfunction}) and rewriting Eq.~(\ref{eq:noisytwopointcorrelationfunction}) as the two-point correlation function for pure state cases. 
The proof can be done by applying the exponential clustering property~\cite{two_myfoot},
which is derived from Lieb-Robinson bounds using the Fourier transformation~\cite{PhysRevB.69.104431,Nachtergaele_2006,hastings2010locality}:
    $|\tr[\rho^{(A+C)}M_A\otimes M_C ]-\tr[\rho^{(A)}M_A]\tr[\rho^{(C)}M_C ]|
    \le c \|M_A\|\|M_C\| \exp{\left(- d(A,C)/\xi\right)}$.
When we choose $M_A=o_A$ and $M_C=(\rho^{(C)})^{n-1}/\mathrm{Tr}[(\rho^{(C)})^n]$, the correlation function is equivalent to $D^{(n)}_0(o_A)$.
For one-dimensional systems, the area law of entanglement entropy $S(\rho_C):=-\tr[\rho^{(C)} \log{\rho^{(C)}}]$ holds; $S(\rho^{(C)}) \le {\rm const}$., as shown in Ref.~\cite{Hastings_2007}.
This implies that the term $\|(\rho^{(C)})^{n-1}\|/\mathrm{Tr}[(\rho^{(C)})^n]$ can be bounded by a constant value independent of $d(A,C)$ for any $n$~\cite{Supplement_localized_virtual_purification}, which leads to $|D^{(n)}_0(o_A)|=O(e^{-d(A,C)})$. For the case of the noisy ground state under the global depolarizing noise channel~\cite{three_myfoot}, the derivation of the inequality and the explicit forms of $\delta_1$ and $\delta_2$ are shown in SM~\cite{Supplement_localized_virtual_purification}.
\end{proof}

\textit{Numerical simulations.---} 
Next, we present the results of numerical simulation to justify our expectation that the LVP is widely valid, even when the conditions for theoretical guarantees do not hold such as the cases 
of a low temperature region
 or gapless (critical) systems~\cite{Supplement_localized_virtual_purification}.
For numerical simulations, we consider the transverse-field Ising (TFI) Hamiltonian on one-dimensional periodic lattice as
\begin{align}
    H_{\rm TFI}=-\sum_{i=1}^N Z_iZ_{i+1} -\lambda \sum_{i=1}^N X_i,
    \label{eq:transverseising}
\end{align}
where $X_i$ and $Z_i$ denote the Pauli-X and Z operators acting on the $i$-th site, and $\lambda$ is the amplitude of the transverse magnetic field.

\begin{figure}[t!]
    \centering
  \includegraphics[width=0.45\textwidth]{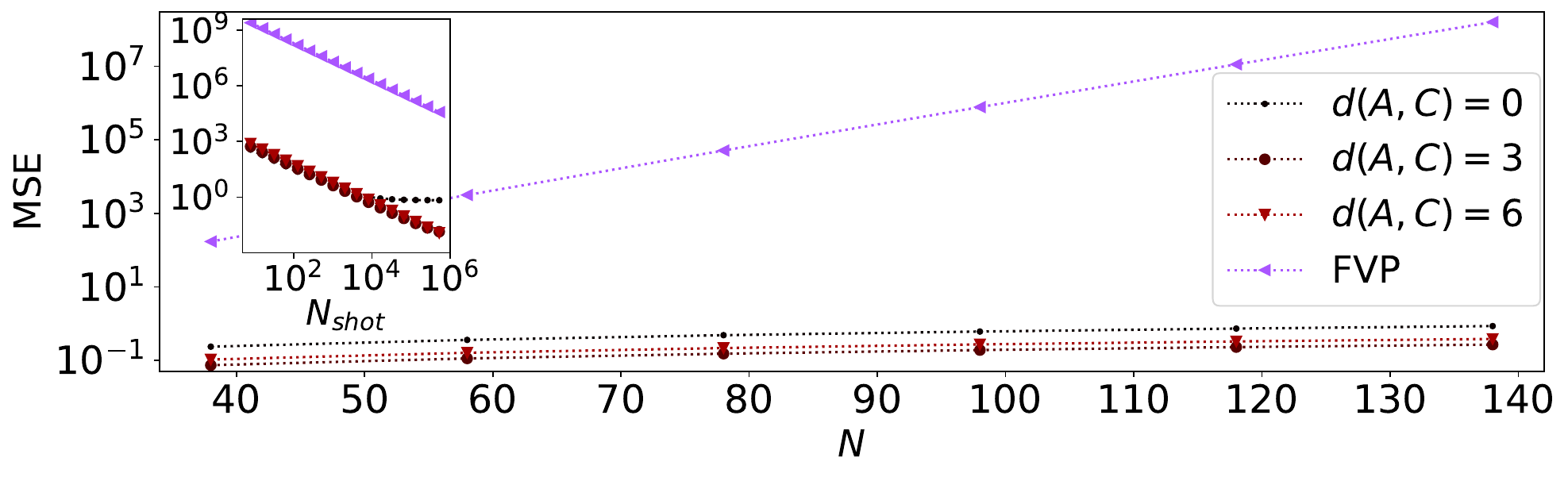}
  \caption{MSE $\chi_{\rm LVP}$ and $\chi_{\rm FVP}$ against the number of qubits $N$ for various $d(A,C)$ at a temperature $\beta=1.0$ and $N_{\rm shot}=2^{14}$. (inset) MSE against $N_{\rm shot}$ for various $d(A,C)$ with $N=138$.}
  \label{RatioMSE_virtualcooling}
\end{figure}
Firstly, we evaluate the performance of the LVP in the context of cooling. As a metric of the performance of LVP, we introduce the mean square error (MSE) $\chi$~\cite{Cai_2021} of the energy of the low-temperature state $\rho_{n \beta}(H_{\rm TFI})$ that is estimated using the high-temperature state $\rho_{\beta}(H_{\rm TFI})$. 
The MSE $\chi$ is composed of two contributions, namely the variance obtained by $N_{\rm shot}$ 
and the bias 
from the true value, 
and also the square root MSE 
can be understood as the estimation precision under $N_{\rm shot}$ measurements.

Figure~\ref{RatioMSE_virtualcooling} shows the MSE $\chi_{\rm LVP}$ and $\chi_{\rm FVP}$
against the number of qubits $N$ for various $d(A,C)$.
We find that 
our LVP protocol only requires the number of measurements proportional to $N$, while the original FVP suffers the exponential number of measurements.
We can also see that $\chi_{\rm FVP},\chi_{\rm LVP} \propto (N_{\rm shot})^{-1}$ widely holds, which implies that the contribution from the bias is negligible unless $N_{\rm shot}$ is taken to be a significantly large number.

Secondly, we verify the performance of the LVP in the context of error mitigation.
Concretely,
we prepare the ground state of $H_{\rm TFI}$ in Eq.~(\ref{eq:transverseising}) at $\lambda=2$ (non-critical) so that two-point correlation functions as well as Eq.~(\ref{Eq:exponentialclustering})  decay exponentially.
We consider the single-qubit local depolarizing noise, often employed to explain the experimental results~\cite{2019Natur.574..505A}, which is defined by
$\mathcal{E}^{(k)}[\rho]=(1-p)\rho+\frac{p}{3}(X_k\rho X_k+Y_k\rho Y_k+Z_k\rho Z_k)$ for the $k$-th qubit, where $p$ is the error rate.

%
%

\begin{figure}[t!]
    \centering
  \includegraphics[width=0.45\textwidth]{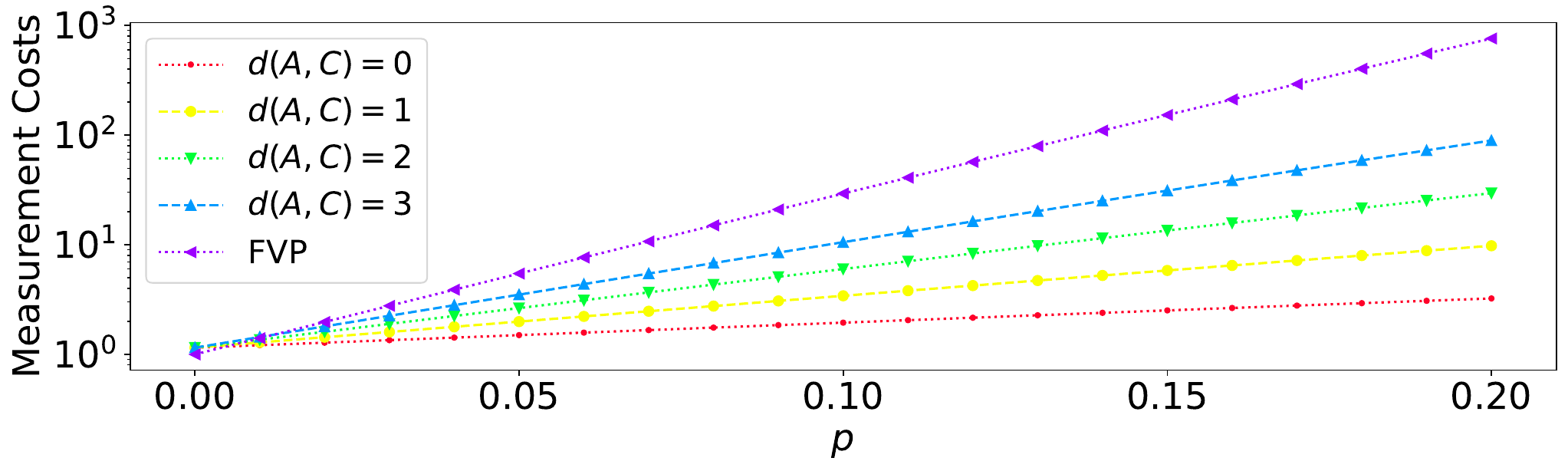}
  \caption{Measurement costs to perform error mitigation by FVP and LVP for a ground state of $H_{\rm TFI}$ in Eq.~(\ref{eq:transverseising}) with $\lambda=2$.
  Here we take $n=2$ copies of noisy ground state for $N=12$ qubits under local depolarizing noise of error rate $p$.}
  \label{Fig:expectationvalues}
\end{figure}
Figure \ref{Fig:expectationvalues} shows the measurement costs to perform the LVP for the noisy ground state of the Hamiltonian $H_{\rm TFI}$ in Eq.~(\ref{eq:transverseising}) at $\lambda=2$ (non-critical).
We can see that the growth of the measurement costs for LVP is significantly slower than that of FVP.
For example, the measurement costs of FVP for $p=0.15$ in Fig.~\ref{Fig:expectationvalues} are more than 100 while those of LVP are around 10, and thus we have a cost reduction by a factor of 10. 
Note that there is a tradeoff between the bias and measurement cost regarding $d(A, C)$~\cite{four_myfoot}.

Finally, we consider a unification of these two protocols: cooling and error mitigation. Namely, we consider a situation that any Gibbs state cannot be prepared perfectly. In order to describe the imperfection, we assume that the input state is a Gibbs state subject to the local depolarizing noise channel on all qubits. 
\begin{figure}[t!]
    \centering
  \includegraphics[width=0.45\textwidth]{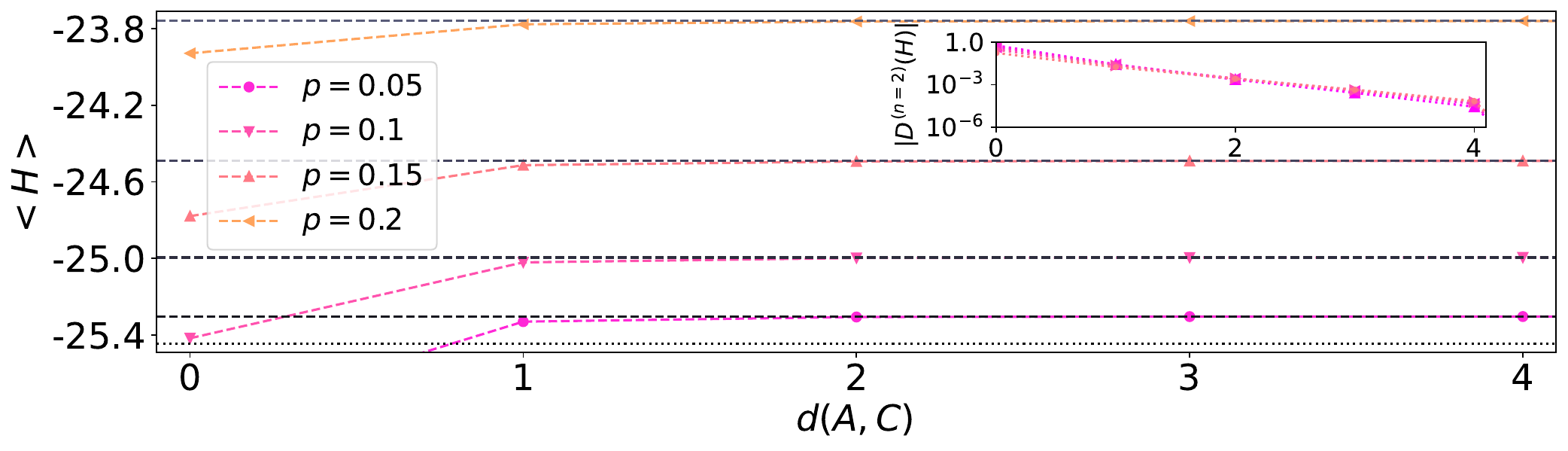}
  \caption{Unifying cooling and error mitigation to simulate  
  $\langle H_{\rm TFI} \rangle$ for $\lambda=2$. The LVP and FVP is performed for $n=2$ copies of noisy Gibbs state at $\beta=1$ for $N=12$ qubits.
  Horizontal dashed and dotted lines 
  indicate the results from the FVP with and without noise, respectively. (inset) Total deviation of the expectation values $D^{(n=2)}(H)$, where we define $|D^{(n)}(H)|$ by $|D^{(n)}(H)| = |\braket{H}^{(n)}_{\rm LVP}-\braket{H}^{(n)}_{\rm FVP}|$, between FVP and LVP against $d(A,C)$.
  }\label{Fig:expectationvalues_virtualcoolingwithnoise}
\end{figure}
Figure~\ref{Fig:expectationvalues_virtualcoolingwithnoise} shows the expectation values of $H_{\rm TFI}$ in Eq.~(\ref{eq:transverseising}) for Gibbs states with local depolarizing noise.
We see that the purification-based approaches are capable of performing both cooling and error mitigation simultaneously, and also that the deviation between LVP and FVP decreases
exponentially with $d(A, C)$ while the measurement costs are also suppressed exponentially.

\textit{Discussion and Conclusion.---} 
In this Letter, we have introduced the Localized Virtual Purification (LVP) that incorporates and utilizes the property of geometrical locality to enhance the utility of purification-based quantum simulation.
We have given theoretical guarantees that the LVP significantly alleviates the severe measurement overhead of the FVP when the two-point correlation function decays exponentially.
This includes the cooling of the Gibbs states of local Hamiltonians 
(Theorem \ref{Theorem1})
and error mitigation of ground state simulation in gapped systems (Theorem \ref{Theorem2}). We have also verified our findings via numerical simulation, and also shown that the LVP is capable of unifying both tasks.

We comment that an additional advantage of our LVP protocol is to alleviate the noise effect among the entangled measurement operations.
A previous study~\cite{vikstål2022study} has shown that the measurement costs of FVP increase exponentially with the number of qubits due to errors in the entangled measurement circuit. 
Using our LVP method, the measurement costs 
do not 
exponentially increase because the controlled derangement circuit is operated only in a local region. 
This is true not only for error mitigation but also for cooling achieved from FVP.

The LVP for the high-temperature regions of Gibbs states in Eq.~(\ref{eq:gibbs}) is one of the most effective cases of our protocol for practical quantum advantages.
For many practical cases, calculating thermodynamic properties for large system sizes of two- or higher-dimensional quantum systems is a challenging task on classical computers even with tensor network methods such as projected entangled-pair states~\cite{PhysRevResearch.1.033038}.
On the other hand, 
there exists an efficient quantum algorithm for preparing thermal states~\cite{brandao2019finite}, under the assumption of the exponential clustering of two-point correlation functions and approximating quantum Markov property, using quantum belief propagation~\cite{brandao2019finite,Kato_2019,PhysRevB.76.201102,PhysRevB.86.245116}. 
That situation completely matches that of our LVP protocol and therefore LVP works on a practical situation using quantum computers.

Last but not least, we comment on other promising applications of the LVP: 
the detection of topological orders~\cite{PhysRevLett.126.050501,doi:10.1126/sciadv.aaz3666}, the measurement of a characterized quantity of quantum chaos~\cite{PhysRevX.12.011018}, and an entanglement measure for mixed states~\cite{PhysRevLett.125.200501}. These quantities are well-known as fundamental concepts in condensed matter physics and quantum information and require to be measured efficiently on quantum computers. By applying our protocol, it may be possible to not only detect such physical quantities but also to incorporate error mitigation or cooling simultaneously. We believe that these are important issues for future research.

\textit{Acknowledgements.---}
We would like to thank Tomotaka Kuwahara, Wataru Mizukami, Yuichiro Yoshida, Yuta Shingu, and Tsuyoshi Okubo for fruitful discussions. 
This work was supported by MEXT Quantum Leap Flagship Program (MEXT Q-LEAP) Grant Number JPMXS0120319794.
This work was supported by JST Grant Number JPMJPF2221.
This
work was also supported by JST Moonshot R\&D (Grant
Number JPMJMS226C).
S.E. is supported by JST PRESTO No. JPMJPR2114.
Y.M. is supported by JSPS KAKENHI (Grant Number 23H04390).
N.Y. wishes to thank JST PRESTO No. JPMJPR2119, and the support from IBM Quantum.
We perform all the numerical calculations with QuTiP \cite{Johansson_2012}.


\let\oldaddcontentsline\addcontentsline
\renewcommand{\addcontentsline}[3]{}
\bibliographystyle{apsrev4-1}
\bibliography{ref}

\begin{thebibliography}{110}
\bibitem [{two()}]{two_myfoot}%
  \BibitemOpen
  \href@noop {} {}\bibinfo {note} {As a remark, there exist other examples
  which show the exponential clustering: excited states for many-body localized
  systems under the assumption of strong dynamical
  localization~\cite{PhysRevLett.114.170505}, the interacting boson systems
  with unbounded local energies~\cite{PhysRevLett.127.070403}, and the
  dynamical systems after the unitary time
  evolution~\cite{PhysRevLett.97.050401} and the non-unitary time evolution
  described by the Lindblad master equation~\cite{PhysRevLett.104.190401}.
  These systems would be a suitable system for our LVP protocol.}\BibitemShut
  {Stop}%
\end{thebibliography}
\let\addcontentsline\oldaddcontentsline
\onecolumngrid

\clearpage
\begin{center}
	\Large
	\textbf{Supplementary Materials for: Localized Virtual purification}
\end{center}

\setcounter{section}{0}
\setcounter{equation}{0}
\setcounter{figure}{0}
\setcounter{table}{0}
\renewcommand{\thesection}{S\arabic{section}}
\renewcommand{\theequation}{S\arabic{equation}}
\renewcommand{\thefigure}{S\arabic{figure}}
\renewcommand{\thetable}{S\arabic{table}}

\addtocontents{toc}{\protect\setcounter{tocdepth}{0}}
{
\hypersetup{linkcolor=blue}
\tableofcontents
}

\section{Derivation of generalized correlation function in Eq.~(\ref{eq:noisytwopointcorrelationfunction})}
We derive that the deviation between the output from the Locally Virtual Purification (LVP) and the Fully Virtual Purification (FVP) can be written as a generalized correlation function as provided in Eq.~\eqref{eq:noisytwopointcorrelationfunction} in the main text. Recall that, given a full-rank density operator $\rho$ in the total system, the deviation of the expectation value of an observable $o_A$ supported on the region $A$ (see Fig.~\ref{Fig:fig1} for a graphical description of regions) is defined as
\begin{align}
    D^{(n)}(o_A)=\frac{\mathrm{Tr}_{A+B}[(\rho^{(A+B)})^n o_A]}{\mathrm{Tr}_{A+B}[(\rho^{(A+B)})^n]}-\frac{\mathrm{Tr}[\rho^n o_A]}{\mathrm{Tr}[\rho^n]}.
    \label{Seq:deviation}
\end{align}
For convenience of discussion, we first define two operators:
\begin{align}
    X_n &=\log{\left(\tr_C[\rho]^n\otimes (\sigma_C)^n \right)}=\log{[(\rho^{(A+B)}\otimes \sigma_C)^n]},\label{Seq:xn}\\
    Y_n &=\log{\left(\rho^n \right)}\label{Seq:yn},
\end{align}
where $\sigma_C$ is an arbitrary positive definite operator defined in region $C$.
By using Eqs.~(\ref{Seq:xn}) and (\ref{Seq:yn}), we can straightforwardly rewrite Eq.~(\ref{Seq:deviation}) as
\begin{align}
    &D^{(n)}(o_A)=\frac{\mathrm{Tr}[e^{X_n} o_A]}{\mathrm{Tr}[e^{X_n}]}-\frac{\mathrm{Tr}[e^{Y_n} o_A]}{\mathrm{Tr}[e^{Y_n}]}.
    \label{Seq:deviation2}
\end{align}
When $\rho$ is not a full-rank operator but still $\rho^{(A+B)}$ is full-rank,
a slightly modified formula (introduced in the following of this section) can be applicable.
As a representative case, we will consider that $\rho$  is a pure state.
If both $\rho^{(A+B)}$ and $\rho$ are pure, this is a trivial case: $D^{(n)}(o_A)=0$.

We also define an operator that interpolates $X_n$ and $Y_n$ using a parameter $\lambda$
($0\le\lambda\le 1$):
\begin{align}
    H_\lambda = Y_n +\lambda(X_n-Y_n).
\end{align}
This allows us to rewrite the deviation as
\begin{align}
    D^{(n)}(o_A)=\frac{\mathrm{Tr}[e^{X_n} o_A]}{\mathrm{Tr}[e^{X_n}]}-\frac{\mathrm{Tr}[e^{Y_n} o_A]}{\mathrm{Tr}[e^{Y_n}]}=\mathrm{Tr}\left(o_A\int_0^1 d\lambda \frac{d}{d\lambda}\frac{e^{ H_\lambda}}{\tr[e^{ H_\lambda}]} \right).
\end{align}
Note that the derivative in the integrand is written explicitly as
\begin{align}
    \frac{d}{d\lambda}\frac{e^{ H_\lambda}}{\tr[e^{ H_\lambda}]}=\frac{1}{\tr[e^{ H_\lambda}]} \frac{d}{d\lambda}e^{ H_\lambda}-\frac{e^{ H_\lambda}}{\tr[e^{ H_\lambda}]^2}\mathrm{Tr}\left(\frac{d}{d\lambda}e^{ H_\lambda} \right).
\end{align}
From the Duhamel's formula (see the derivation in Ref.~\cite{haber2018notes})
\begin{align}
    \frac{d}{d\lambda}e^{ H_\lambda}= \int_0^1 (e^{ H_\lambda})^\tau \left(\frac{d}{d\lambda}H_\lambda \right) (e^{ H_\lambda})^{1-\tau} d\tau, 
\end{align}
we finally obtain Eq.~\eqref{eq:noisytwopointcorrelationfunction} in the main text as
\begin{align}
    D^{(n)}(o_A)&= \mathrm{Tr}\left(o_A\int_0^1 d\lambda\int_0^1 d\tau \rho_\lambda^\tau (X_n-Y_n)\rho_\lambda^{1-\tau} - \rho_\lambda\mathrm{Tr}[\rho_\lambda^\tau (X_n-Y_n)\rho_\lambda^{1-\tau}] \right)\\
    &= \int_0^1 d\lambda\int_0^1 d\tau ~\mathrm{Corr}_{\rho_\lambda}^\tau (X_n-Y_n,o_A),
    \label{Seq:generalizedcorr1}
\end{align}
where we have introduced $\rho_\lambda=\frac{e^{ H_\lambda}}{\tr[e^{ H_\lambda}]} $ and a generalized correlation function as
\begin{align}
    \mathrm{Corr}_\rho^\tau (O,O') =\mathrm{Tr}[\rho^\tau O \rho^{1-\tau}O'] - \mathrm{Tr}[\rho O]\mathrm{Tr}[\rho O'].
\end{align}

If the state of the total system $\rho$ is pure, a slightly modified version of Eq.~(\ref{Seq:generalizedcorr1}) can be obtained as
\begin{align}
    D^{(n)}(o_A)&= \int_0^1 d\lambda\int_0^1 d\tau ~\mathrm{Corr}_{\rho^{(A+B)}_\lambda}^\tau (X^{(A+B)}_n-Y^{(A+B)}_n,o_A)
    \label{Seq:generalizedcorr2},
\end{align}
where we define
\begin{align}
    X^{(A+B)}_n &=\log{\left(\tr_C[\rho]^n \right)}=\log{[(\rho^{(A+B)})^n]},\label{Seq:xn2}\\
    Y^{(A+B)}_n &=\log{\tr_C\left[\rho^n \right]}\label{Seq:yn2},
\end{align}
and
\begin{align}
    H^{(A+B)}_\lambda &= Y^{(A+B)}_n +\lambda(X^{(A+B)}_n-Y^{(A+B)}_n),\\
    \rho^{(A+B)}_\lambda&=\frac{e^{ H^{(A+B)}_\lambda}}{\tr_{A+B}[e^{ H^{(A+B)}_\lambda}]}.
\end{align}

\subsection{Gibbs states $\rho_{\beta}(H)$: Derivation of $X_n$ and $Y_n$}
In this subsection, we derive an explicit form of $X_n$ and $Y_n$ assuming that the total quantum state $\rho$ is a canonical Gibbs state
\begin{align}
    \rho_{\beta}(H) = \frac{e^{-\beta H}}{\tr[e^{-\beta H}]}\notag,
\end{align}
for an inverse temperature $\beta$ under Hamiltonian $H$, as defined in Eq.~(\ref{eq:gibbs}) in the main text.
By definition, $Y_n$ can be directly calculated from Eq.~\eqref{Seq:yn} as
\begin{align}
Y_n =-n\beta H.
\end{align}
Here, we omit a constant term because the numerator and denominator in Eq.~(\ref{Seq:deviation2}) cancel each other out.
Furthermore, we choose $\sigma_C = e^{-\beta H^{(C)}}$, where $H^{(C)}$ 
are defined as the terms in original Hamiltonian $H$ that are solely supported on the region $C$.
Then, from Eq.~\eqref{Seq:xn} we obtain 
\begin{align}
X_n =-n\beta (H+\Delta H).
\end{align}
Here, the definition of $\Delta H$ is 
\begin{align}
    \Delta H = -H^{(B C)} +\Delta H^{(A+B)},
    \label{Seq:defdeltaH}
\end{align}
where $H^{(B C)}$ is the interaction Hamiltonian defined as the terms in original Hamiltonian $H$ that are  supported on both of the region $B$ and $C$, and the definition of $\Delta H^{(A+B)}$ is given by
\begin{align}
    \Delta H^{(A+B)}= -\frac{1}{\beta}\log\tr_C[\rho_{\beta}(H)] - H^{(A+B)},
    \label{Seq:defdeltaHab}
\end{align}
where $H^{(A+B)}$ is the Hamiltonian defined as the terms in original Hamiltonian $H$ that are  supported on either of the region $A$ and $B$.
$\Delta H^{(A+B)}$ has a meaning of the deviation of the Hamiltonian from a Gibbs state $\rho_{\beta}(H^{(A+B)})$ truncated on a local region $A+B$, by rewriting Eq.~(\ref{Seq:defdeltaHab}) as $e^{-\beta(H^{(A+B)}+\Delta H^{(A+B)}) } = \tr_C[e^{-\beta H}]$. Here, we neglect a constant term of this equality inside the Hamiltonian.
In general, $\Delta H^{(A+B)}$ may be a non-local operator inside the whole region of $A$ and $B$, but $\Delta H^{(A+B)}$ 
can be approximated by a sum of local operators supported on the boundary region of $B$ with exponentially small error.

\subsection{Pure states $\rho_0$: Derivation of $D^{(n)}_0(o_A)$}
In this subsection, we derive an explicit form of $D^{(n)}_0(o_A)$ for the case of a pure state $\rho_0$ from Eq.~(\ref{Seq:generalizedcorr2}).
Assuming that the total quantum state is pure $\rho_0$ so that $(\rho_0)^n =\rho_0$,
$X^{(A+B)}_n$ and $Y^{(A+B)}_n$ can be rewritten with a simpler form as
\begin{align}
    X^{(A+B)}_n &=n\log{[(\rho^{(A+B)})]}=nX^{(A+B)}_1,\\
    Y^{(A+B)}_n &=\log{\tr_C\left[(\rho_0)^n \right]}=Y^{(A+B)}_1=X^{(A+B)}_1,
\end{align}
and also we obtain
\begin{align}
    H^{(A+B)}_\lambda &=(1+(n-1)\lambda)X^{(A+B)}_1,\\
    \rho^{(A+B)}_\lambda&=\frac{(\rho^{(A+B)})^{(1+(n-1)\lambda)}}{\tr_{A+B}[(\rho^{(A+B)})^{(1+(n-1)\lambda)}]}.
\end{align}
By substituting these equations into  Eq.~(\ref{Seq:generalizedcorr2}), we obtain
\begin{align}
    D_0^{(n)}(o_A)&= \int_0^1 d\lambda \tr_{A+B}[\rho^{(A+B)}_\lambda (X^{(A+B)}_n-Y^{(A+B)}_n) o_A]-\tr_{A+B}[\rho^{(A+B)}_\lambda  o_A]\tr_{A+B}[\rho^{(A+B)}_\lambda (X^{(A+B)}_n-Y^{(A+B)}_n) ]\\
    &= (n-1)\int_0^1 d\lambda \frac{\tr_{A+B}[(\rho^{(A+B)})^{(1+(n-1)\lambda)}  X^{(A+B)}_1 o_A]-\tr_{A+B}[\rho^{(A+B)}_\lambda  o_A]\tr_{A+B}[(\rho^{(A+B)})^{(1+(n-1)\lambda)}  X^{(A+B)}_1 ]}{\tr_{A+B}[(\rho^{(A+B)})^{(1+(n-1)\lambda)}]}\label{Seq:generalizedcorrpurebefore}\\
    &= (n-1)\Big{(}
    \int_0^1 d\lambda \frac{\tr_{A+C}[\rho^{(A+C)}   o_A\otimes (\rho^{(C)})^{(n-1)\lambda}\log{[\rho^{(C)}]}]}{\tr_{C}[(\rho^{(C)})^{(1+(n-1)\lambda)}]}\notag\\
    &\qquad\qquad\qquad -\frac{\tr_{A+C}[\rho^{(A+C)}   o_A\otimes (\rho^{(C)})^{(n-1)\lambda}]\tr_{C}[(\rho^{(C)})^{(1+(n-1)\lambda)}\log{[\rho^{(C)}]} ]}{\tr_{C}[(\rho^{(C)})^{(1+(n-1)\lambda)}]^2} \Big{)}.
    \label{Seq:generalizedcorrpure}
\end{align}
Note that, from Eq.~(\ref{Seq:generalizedcorrpurebefore}) to (\ref{Seq:generalizedcorrpure}), we have employed the following relation that can be easily proven using the Schmidt decomposition for the pure state $\rho_0$ as
\begin{align}
    \tr_{A+B}[(\rho^{(A+B)})^{(1+(n-1)\lambda)}  \log{[\rho^{(A+B)}]} o_A]&=\tr_{A+B+C}[\rho_0 (\rho^{(A+B)})^{(n-1)\lambda}  \log{[\rho^{(A+B)}]} o_A]\\
    &=\tr_{A+B+C}[\rho_0 (\rho^{(C)})^{(n-1)\lambda}  \log{[\rho^{(C)}]} o_A]\\
    &=\tr_{A+C}[\rho^{(A+C)} o_A\otimes (\rho^{(C)})^{(n-1)\lambda}  \log{[\rho^{(C)}]} ].
\end{align}
After the integration of Eq.~(\ref{Seq:generalizedcorrpure}), we obtain the result
\begin{align}
    D^{(n)}_0(o_A)=\mathrm{Tr}_{\rm A+C}\left[(\rho^{(A+C)}-\rho^{(A)}\otimes \rho^{(C)})o_A \otimes\frac{(\rho^{(C)})^{n-1}}{\mathrm{Tr}[(\rho^{(C)})^n]}\right].
\end{align}

\section{Formalism of Fully and localized virtual purification}
In this section, we explain the formalism of FVP and LVP,
 quantum circuits for evaluating the expectation values of FVP and LVP, and also their measurement costs.
Let us first describe the FVP, which 
is a protocol to obtain expectation values of the purified quantum states $\rho_{\rm FVP}=\rho^n/ \tr{[\rho^n]}$ by three steps: preparing $n$ copies of $\rho$, coherent operations among copies, and then the post-processing procedure
%
\cite{PhysRevX.11.031057,PhysRevX.11.041036,https://doi.org/10.48550/arxiv.2210.10799,PhysRevLett.129.020502}.
A quantum circuit as shown in Fig.~\ref{circuit:virtualdistillation}
can evaluate the expectation values of a Pauli operator $P_a$ within the formalism of the FVP:
\begin{align}
    \braket{P_a}_{\rm FVP}=\tr[P_a\rho_{\rm FVP}]&=\frac{\tr[P_a \rho^n]}{\tr[ \rho^n]}\notag\\&=\frac{\mathrm{Tr}[ (\rho^{\otimes n}) \mathbb{S}(P_a\otimes \mathbb{I}^{\otimes n-1})]}{\mathrm{Tr}[(\rho^{\otimes n})\mathbb{S}]}.
    \label{eq:virtualdef}
\end{align}
Here, $\mathbb{S}$ is the cyclic-shift operator (or called the derangement operator, which can be decomposed of a product of the SWAP operators) between $n$ copies, defined by $\mathbb{S}\ket{\psi_1\psi_2\cdots \psi_n}=\ket{\psi_n\psi_1\psi_2\cdots \psi_{n-1}} $, and $\mathbb{I}$ is the identity operator.
Note that we can also measure these expectation values without any ancillary qubits by using the Bell-like collective measurements for a multi-copy case~\cite{PhysRevX.11.041036}.
\begin{figure*}[htbp]
\begin{minipage}[b]{0.41\linewidth}
\begin{align*}
\Qcircuit @C=1.2em @R=1.2em {
&\lstick{\ket{0}}&\gate{H}&\ctrl{12}&\ctrl{11}&\ctrl{10}&\ctrl{9}&\ctrl{8}&\ctrl{7}&\ctrl{1} &\measureD{X}\\
& & \qw& \qw& \qw & \qw& \qw& \qw& \qswap& \gate{P_a}& \qw \\
& & \qw& \qw& \qw & \qw& \qw& \qswap& \qw& \qw& \qw \\
& & \qw& \qw& \qw & \qw& \qswap& \qw& \qw& \qw& \qw \\
& & \qw& \qw& \qw & \qswap& \qw& \qw& \qw& \qw& \qw \\
& & \qw& \qw& \qswap & \qw& \qw& \qw& \qw& \qw& \qw \\
& \lstick{\raisebox{6.6em}{$\rho$\ }}& \qw& \qswap\gategroup{2}{2}{7}{5}{.3em}{\{}& \qw&\qw& \qw & \qw & \qw& \qw& \qw \\
& & \qw& \qw& \qw & \qw& \qw& \qw& \qswap& \qw& \qw \\
& & \qw& \qw& \qw & \qw& \qw& \qswap& \qw& \qw& \qw \\
& & \qw& \qw& \qw & \qw& \qswap& \qw& \qw& \qw& \qw \\
& & \qw& \qw& \qw & \qswap& \qw& \qw& \qw& \qw& \qw \\
& & \qw& \qw& \qswap & \qw& \qw& \qw& \qw& \qw& \qw \\
& \lstick{\raisebox{5.6em}{$\rho$\ }}& \qw& \qswap\gategroup{8}{2}{13}{5}{.5em}{\{}& \qw&\qw& \qw & \qw & \qw&\qw&\qw \\
}
\end{align*}
\caption{A quantum circuit for evaluating the expectation value of $P_a$ based on the FVP method as described in Eq.~(\ref{eq:virtualdef}). The numerator can be estimated as $ \tr [\rho^n P_a]$, where we set $n=2$ copies and $N=6$ qubits. 
The denominator $\tr[\rho^n]$ can be evaluated if we eliminate the controlled-$P_a$ gate.
}
\label{circuit:virtualdistillation}
\end{minipage}
\begin{minipage}[b]{0.41\linewidth}
\begin{align*}
\Qcircuit @C=1.2em @R=1.2em {
&\lstick{\ket{0}}&\gate{H}&\qw&\ctrl{9}&\ctrl{8}&\ctrl{7}&\ctrl{2} &\measureD{X}\\
& & \qw & \qw& \qw& \qw& \qswap& \qw& \qw \\
& & \qw & \qw& \qw& \qswap& \qw& \gate{P_a}& \qw \\
& & \qw & \qw& \qswap& \qw& \qw& \qw& \qw \\
& & \qw & \qw& \qw& \qw& \qw& \qw& \qw \\
& & \qw & \qw& \qw& \qw& \qw& \qw& \qw \\
& \lstick{\raisebox{6.6em}{$\rho$\ }}& \qw& \qw\gategroup{2}{2}{7}{5}{.3em}{\{}& \qw&\qw& \qw & \qw & \qw \\
& & \qw & \qw& \qw& \qw& \qswap& \qw& \qw \\
& & \qw & \qw& \qw& \qswap& \qw& \qw& \qw \\
& & \qw & \qw& \qswap& \qw& \qw& \qw& \qw \\
& & \qw & \qw& \qw& \qw& \qw& \qw& \qw \\
& & \qw & \qw& \qw& \qw& \qw& \qw& \qw \\
& \lstick{\raisebox{5.6em}{$\rho$\ }}& \qw& \qw\gategroup{8}{2}{13}{5}{.5em}{\{}& \qw&\qw& \qw & \qw & \qw \\
}
\end{align*}
\caption{A quantum circuit for evaluating the expectation value of $P_a$ based on the LVP method as described in Eq.~(\ref{eq:partialdef}). The numerator can be estimated as $\tr [\rho^n P_a]$, where we set $n=2$ copies and $N=6$ qubits ($N_A=1$, $N_B=2$, $N_C=3$, and $d(A,C)=1$). We perform the controlled cyclic-shift operators locally around the support of $P_a$, in contrast to the FVP circuit in Fig.~\ref{circuit:virtualdistillation}.
}
\label{circuit:partialdistillation}
\end{minipage}
\end{figure*}
While the FVP requires the globally entangled measurement that involves the entire system among the $n$ copies, the LVP only requires entangled measurement to involve the local subset that covers the support of the observable.
Figure~\ref{circuit:partialdistillation} shows a quantum circuit for evaluating the expectation values of LVP:
\begin{align}
    \frac{\tr_{A+B}[P_a (\rho^{(A+B)})^n]}{\tr_{A+B}[ (\rho^{(A+B)})^n]}&=\frac{\mathrm{Tr}[ (\rho^{\otimes n}) \mathbb{S}^{(A+B)}(P_a\otimes \mathbb{I}^{\otimes n-1})]}{\mathrm{Tr}[(\rho^{\otimes n})\mathbb{S}^{(A+B)}]}.
    \label{eq:partialdef}
\end{align}
Here, $\mathbb{S}^{(A+B)}$ is a derangement operator that acts on the regions $A$ and $B$.
After summing up contributions from different $P_a$'s,
we can obtain the expectation values of LVP in Eq.~(\ref{eq:partialvirtual}). 

The above discussion assumes that we prepare $n$ copies of a quantum state $\rho$ and also operate controlled derangement operations.
Meanwhile, it is in principle possible to perform the FVP and LVP only with a single-copy measurement without entangled measurement operation, using techniques such as the classical shadow tomography~\cite{Huang_2020,PRXQuantum.4.010303}. Note that this is in compensation for the exponentially large number of measurements~\cite{doi:10.1126/science.abn7293}.
On the other hand,
the number of measurements required for our LVP method depends on $N_A+N_B$, but not on $N$ (the measurement costs for FVP and LVP using the controlled derangement circuits are shown in the following subsection).
We also mention that there exists a hybrid approach of multiple-copy and single-copy measurement, which contains a trade-off relation between the number of copies and the measurement overhead~\cite{zhou2022hybrid}.

\subsection{Variance for estimating the expectation values of FVP and LVP}\label{App:costandvariance}
In this subsection, we compute the variance for estimating the expectation values of FVP and LVP by using the quantum circuits in Figs.~\ref{circuit:virtualdistillation} and \ref{circuit:partialdistillation}.
For simplicity, we assume that the denominators and the numerators for FVP and LVP are independently measured.
If we estimate $E_1$ ($E_2$) with a statistical error $\delta E_1$ ($\delta E_2$), then an estimated value of $E_1/E_2$ is given by
\begin{align}
    \frac{E_1+\delta E_1}{E_2+\delta E_2}= \frac{E_1}{E_2}+\frac{\delta E_1}{E_2}-\frac{E_1\delta E_2}{E_2^2} +O(\delta^2),
\end{align}
and we obtain its variance 
\begin{align}
    {\rm Var}\left(\frac{E_1}{E_2}\right)&\simeq \frac{{\rm Var}(E_1)}{(E_2)^2}+\frac{(E_1)^2{\rm Var}(E_2)}{(E_2)^4}-2E_1\frac{{\rm Cov}(E_1,E_2)}{(E_2)^3}\label{appeq:varianceformulabefore}\\
    &=\frac{{\rm Var}(E_1)}{(E_2)^2}+\frac{(E_1)^2{\rm Var}(E_2)}{(E_2)^4}.
    \label{appeq:varianceformula}
\end{align}
Here, ${\rm Var}(\bullet)$ is a variance and ${\rm Cov}(E_1,E_2)$ is a covariance of $E_1$ and $E_2$.
From Eq.~(\ref{appeq:varianceformulabefore}) to (\ref{appeq:varianceformula}), we used ${\rm Cov}(E_1,E_2)=0$, which follows from the assumption that $E_1$ and $E_2$ are estimated independently.

\subsubsection{Fully Virtual Purification}
In the case of the FVP as shown in Fig.~\ref{circuit:virtualdistillation}, we substitute $E_1=\tr[P_a \rho^n]$, $E_2=\tr[\rho^n]$, ${\rm Var}(E_1)=1-\tr[P_a \rho^n]^2$, and ${\rm Var}(E_2)=1-\tr[\rho^n]^2$ into Eq.~(\ref{appeq:varianceformula}) to obtain
\begin{align}
    {\rm Var}(P_a)_{\rm FVP} = \frac{1-\tr[P_a \rho^n]^2}{\tr[ \rho^n]^2}+\frac{\tr[P_a \rho^n]^2(1-\tr[\rho^n]^2)}{\tr[ \rho^n]^4}.
    \label{appeq:variance_FVP}
\end{align}
Moreover, the variance can be bounded as a simpler form using the relation $\tr[P_a \rho^n]\le \tr[\rho^n]\le 1$ and $\tr[P_a\rho^n]\le 1$:
\begin{align}
    {\rm Var}(P_a)_{\rm FVP} \le \frac{1}{\tr[ \rho^n]^2}+\frac{\tr[ \rho^n]^2}{\tr[ \rho^n]^4}=\frac{2}{\tr[ \rho^n]^2}.
\end{align}
This implies that  $\tr[ \rho^n]^{-2}$ can be regarded as the measurement cost to perform the FVP.

In particular, if $\rho$ is a Gibbs state, then we obtain
\begin{align}
    \tr[ (\rho_\beta(H))^n]^{-2}=\exp{[2n\beta N(f_{n\beta}-f_{\beta})]},\label{eq:fvp_cooling_cost}
\end{align}
where $f_{\beta}:=-\frac{1}{N\beta}\log\tr[e^{-\beta H}]$ is a free energy density of the Hamiltonian $H$ at the inverse temperature $\beta$.
Note that the difference of the free energy density is positive ($f_{n\beta}-f_{\beta}>0$) because of the fact that $\frac{\partial f_{\beta}}{\partial \beta} = \frac{1}{\beta^2} s_\beta >0$, where $s_\beta>0$ is the thermodynamic entropy density at the inverse temperature $\beta$. Therefore, we can see that measurement costs exponentially increase as the system size $N$ increases.

\subsubsection{Localized Virtual Purification}
In the case of the LVP as shown in Fig.~\ref{circuit:partialdistillation}, we substitute 
$E_1=\tr_{A+B}\left[P_a \left(\rho^{(A+B)}\right)^n\right]$, $E_2=\tr_{A+B}\left[\left(\rho^{(A+B)}\right)^n\right]$, ${\rm Var}(E_1)=1-\tr_{A+B}\left[P_a \left(\rho^{(A+B)}\right)^n\right]^2$, and ${\rm Var}(E_2)=1-\tr_{A+B}\left[ \left(\rho^{(A+B)}\right)^n\right]^2$
into Eq.~(\ref{appeq:varianceformula}) to obtain
\begin{align}
    {\rm Var}(P_a)_{\rm LVP} = \frac{1-\tr_{A+B}\left[P_a \left(\rho^{(A+B)}\right)^n\right]^2}{\tr_{A+B}\left[ \left(\rho^{(A+B)}\right)^n\right]^2}+\frac{\tr_{A+B}\left[P_a \left(\rho^{(A+B)}\right)^n\right]^2\left(1-\tr_{A+B}\left[\left(\rho^{(A+B)}\right)^n\right]^2\right)}{\tr_{A+B}\left[\left(\rho^{(A+B)}\right)^n\right]^4}.
    \label{appeq:variance_LVP}
\end{align}
Similarly, we obtain
\begin{align}
    {\rm Var}(P_a)_{\rm LVP} \le \frac{2}{\tr_{A+B}\left[\left(\rho^{(A+B)}\right)^n\right]^2},
\end{align}
where $\tr_{A+B}\left[\left(\rho^{(A+B)}\right)^n\right]^{-2}$ can be regarded as the measurement cost. 

In particular, if $\rho$ is a Gibbs state, $\rho^{(A+B)}$ can be approximated by a Gibbs state defined by the truncated Hamiltonian $H^{(A+B)}$ supported solely on the regions $A$ and $B$.
Then we obtain
\begin{align}
    \tr_{A+B}\left[\left(\rho^{(A+B)}\right)^n\right]^{-2}\simeq \tr_{A+B}\left[\left(\rho_\beta(H^{(A+B)})\right)^n\right]^{-2}=\exp{[2n\beta(N_A+N_B)(f_{n\beta}-f_{\beta})]},
\end{align}
where $f_{\beta}$ is a free energy density of the Hamiltonian $H^{(A+B)}$.
By comparing this with Eq.~\eqref{eq:fvp_cooling_cost}, we can see that the measurement costs are exponentially small regarding $N_C$ compared with the FVP:
\begin{align}
    \frac{\tr_{A+B}\left[\left(\rho^{(A+B)}\right)^n\right]^{-2}}{\tr[ \rho_\beta(H)^n]^{-2}}\simeq \exp{[-2n\beta N_C(f_{n\beta}-f_{\beta})]}.
    \label{appeq:cost_derivation}
\end{align}
This approximation relies on a natural assumption that the subsystem states are in the same thermal equilibrium states of the total system, namely providing the same thermodynamic functions, such as the free energy, and the expectation values of local observables. 
This assumption is one of the fundamental thermodynamic principles and several studies derive the assumption from purely quantum mechanics~\cite{PhysRevX.4.031019,Fr_hlich_2015,Popescu_2006}.
This approximation finds support in Fig.~2 in the main text, where the exponentially decaying MSE serves as numerical evidence. 
Additionally, we present numerical simulations involving the one-dimensional transverse field Ising model to offer direct validation.
\begin{figure}[htbp]
    \centering
  \includegraphics[width=0.95\textwidth]{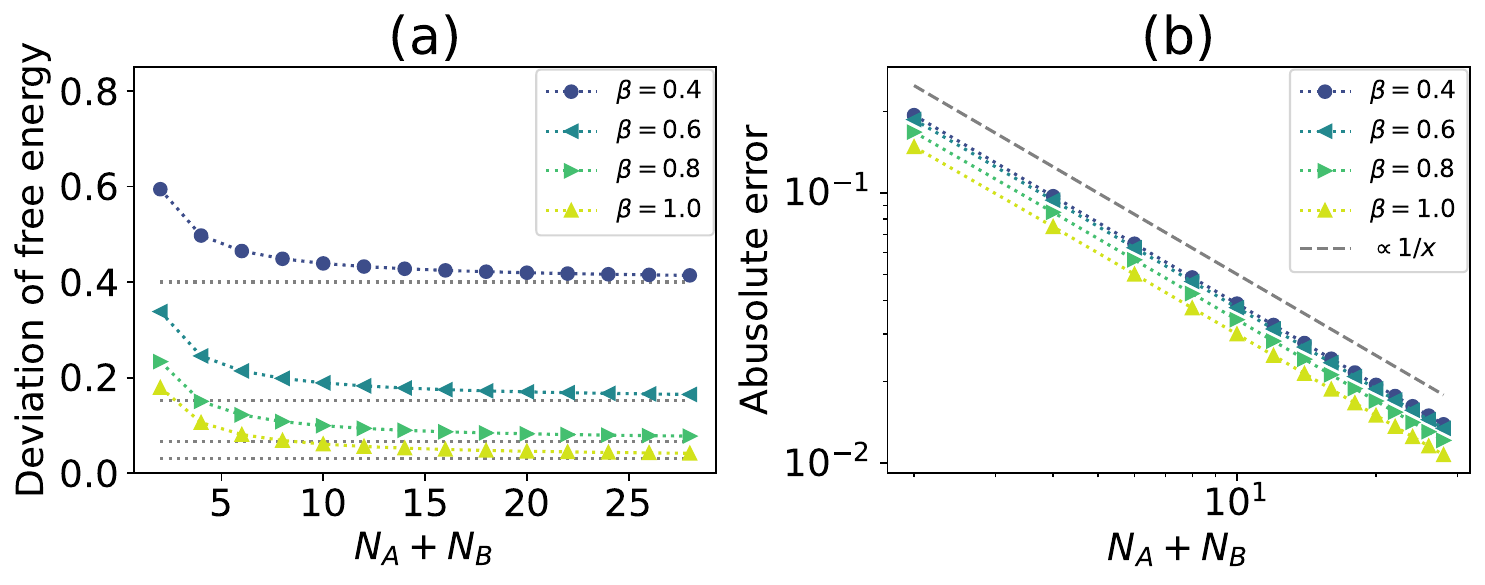}
  \caption{(a) Plots of $-\frac{1}{\beta(N_A+N_B)}\log{\tr_{A+B}\left[\left(\rho^{(A+B)}\right)^2\right]}$ against $N_A+N_B$ for various temperatures. We consider the one-dimensional transverse field Ising model with $\lambda=2$ for $n=2$ copies and take the limit of $N_C \rightarrow \infty$, to omit the finite-size effect of the total system. Horizontal dotted lines denote the values of $(f_{\beta}-f_{2\beta})$ in the thermodynamic limit. (b) Plots of absolute errors $-\frac{1}{\beta(N_A+N_B)}\log{\tr_{A+B}\left[\left(\rho^{(A+B)}\right)^2\right]}-(f_{\beta}-f_{2\beta})$ against $N_A+N_B$.}
  \label{Fig:local_purity}
\end{figure}
Figure \ref{Fig:local_purity} (a) shows $-\frac{1}{\beta(N_A+N_B)}\log{\tr_{A+B}\left[\left(\rho^{(A+B)}\right)^2\right]}$ against $N_A+N_B$, indicating that the deviation decays as $N_A+N_B$ increases.
Figure \ref{Fig:local_purity} (b) additionally plots the absolute errors $-\frac{1}{\beta(N_A+N_B)}\log{\tr_{A+B}\left[\left(\rho^{(A+B)}\right)^2\right]}-(f_{\beta}-f_{2\beta})$ against $N_A+N_B$. We can see that the absolute errors behave in proportion to $1/(N_A+N_B)$, resulting in
\begin{align}
    \tr_{A+B}\left[\left(\rho^{(A+B)}\right)^2\right]^{-2}=\exp{[4\beta(N_A+N_B)(f_{2\beta}-f_{\beta})+O(1)]}.
\end{align}
In Fig.~\ref{Fig:local_purity}, we employ the analytical results of the free energy density in the thermodynamic limit as
\begin{align}
    f_\beta=-\frac{1}{2\pi\beta} \int_{-\pi}^\pi dq \log{(e^{\beta \varepsilon_q}+e^{-\beta \varepsilon_q})}
    ,
\end{align}
where we omit the constant values independent of $\beta$ and $\varepsilon_q=\sqrt{(\cos{q}+\lambda)^2+(\sin{q})^2}$ defined in Eq.~(\ref{appeq:singleenergyTFIM}).



\subsection{Examples of tasks: cooling and error mitigation}
We review FVP as a task for cooling and error mitigation, which is referred to as virtual cooling~\cite{PhysRevX.9.031013} and virtual distillation~\cite{PhysRevX.11.031057,PhysRevX.11.041036,https://doi.org/10.48550/arxiv.2210.10799,PhysRevLett.129.020502}.
For cooling, we prepare a canonical Gibbs state $\rho_\beta(H)= e^{-\beta H} / {\rm Tr}[e^{-\beta H}]$ for inverse temperature $\beta$ under Hamiltonian $H$, 
and the density matrix of FVP is given by
\begin{align}
    \rho_{\rm FVP}&=\frac{(\rho_\beta (H))^n}{ \tr{[(\rho_\beta (H))^n]}}
    =\rho_{n\beta}(H).
\end{align}
This means that it is possible to obtain results at inverse temperatures $n\beta$ from a canonical Gibbs state at an inverse temperature $\beta$. 
It is noteworthy that FVP for cooling was experimentally realized using ultracold atoms in optical lattices in Ref.~\cite{PhysRevX.9.031013}.

For error mitigation, we explain the mechanism of mitigating errors as follows:
We consider a noisy state $\rho$ such that its decomposition is given as
\begin{align}
    \rho = (1-\lambda_0) \ket{\psi}\bra{\psi}+ \lambda_0\sum_kp_k \ket{\psi_k}\bra{\psi_k}.
\end{align}
Here, $\ket{\psi}$ is a dominant eigenvector that is known to be practically close to the noiseless state~\cite{Koczor_2021}
and $\ket{\psi}$ and $\ket{\psi_k}$ for all $k$ are orthogonal. $\lambda_0$ and $p_k$ denote the effects of  error and satisfy $ \sum_k p_k = 1$ ($p_k\ge 0$ for all $k$).
Then, the error-mitigated state $\rho_{\rm FVP}$ is exponentially close to the 
dominant eigenvector
in terms of $n$:
\begin{align}
    \tr[\rho_{\rm FVP}O]= \bra{\psi}O\ket{\psi}+O\left(\left(\frac{\lambda_0}{1-\lambda_0}\right)^n\right).
\end{align}
Similarly, FVP for error mitigation was experimentally realized using the superconducting qubits~\cite{https://doi.org/10.48550/arxiv.2210.10799}.

As we have shown in the main text, the LVP can also be applicable for cooling and error mitigation, as long as the deviation $D^{(n)}(o_A)$ is sufficiently small, and therefore its condition is essential for the practical use of LVP.

\section{Formal statement of Theorem \ref{Theorem1} and its proof: exponentially small deviation between FVP and LVP for Gibbs states}\label{sec_devVC}
In this section, we introduce the formal version of Theorem \ref{Theorem1} and present its proof.
First, we introduce an assumption regarding the effective Hamiltonian $\Delta H^{(A+B)}$  in Eq.~(\ref{Seq:defdeltaHab}):
\begin{assumption}\label{assumption1}
For sufficiently small $\beta$, the effective Hamiltonian $\Delta H^{(A+B)}$ exist, satisfying the following two properties:
\begin{enumerate}
\item[(i)] $\Delta H^{(A+B)}$ can be approximated to $\Delta H^{(A+B)}_{\rm local}$ with its support included in the boundary region of $A\cup B$, $|\partial (A\cup B)_l|=|\{v\in (A\cup B)|\min_{j\in C}d_{v,j}\le l\}|$ with exponentially small errors:
\begin{align}
    \|\Delta H^{(A+B)}-\Delta H^{(A+B)}_{\rm local}\|\le c_1|\partial (A\cup B)_l| \exp{[-C_\beta l/r]}\quad (=O(\exp{[-l]})) 
\end{align}
where 
$c_1$ and $C_\beta$ are positive constants and
Hamiltonian has $k$-body interaction with $r$ the length of the interaction.
\item[(ii)] $\Delta H^{(A+B)}_{\rm local}$ can be decomposed into a sum of local operators and the support of $\Delta H^{(A+B)}_{\rm local}$ for each term is equal to or less than $l$.
\end{enumerate}
\end{assumption}
A previous study~\cite{PhysRevLett.124.220601} has proved Assumption \ref{assumption1} for general cases of non-commuting Hamiltonians for the high temperature region above the threshold temperature $\beta_c=\frac{1}{8e^3 kg}$. However, it has turned out that the proof contains a flaw and requires a significant modification~\cite{kuwahara2024qip}. Although Assumption \ref{assumption1} (both (i) and (ii)) holds in specific systems such as commuting Hamiltonians for high temperature region~\cite{bluhm2024strong} and numerical simulations stand for the validity even for a non-commuting  Hamiltonian~\cite{PhysRevB.81.054106}, the general non-commuting cases are still an open problem~\cite{five_myfoot}.

Under Assumption \ref{assumption1}, we can prove the following theorem:
\begin{theorem}\label{Theorem3}(formal version of Theorem \ref{Theorem1})
Assuming the exponential clustering of the two-point correlation function of canonical Gibbs states at an inverse temperature $n\beta$, $|\braket{O_AO_C}_{n\beta}-\braket{O_A}_{n\beta}\braket{O_C}_{n\beta}|\le c'\|O_A\|\|O_C\|e^{-d(A,C)/\xi'}$~\cite{araki1969gibbs,1995JSP....80..223P,ueltschi2005cluster,PhysRevX.4.031019,Fr_hlich_2015}, where $c'$ and $\xi'$ are positive parameters independent of $N$,
the difference $|D^{(n)}(o_A)|$ is exponentially small in terms of $d(A,C)$:
\begin{align}
    |D^{(n)}(o_A)|&\le  c'n\beta\|o_A\|\|\Delta H_{\rm local}\|e^{-d(A,C)/\xi'}+C_\beta' \min{(|\partial A|,|\partial C|)}\|o_A\|\|\Delta H_{\rm local}\|e^{-d(A,C)/\xi_\beta'}\notag\\
    &\qquad +2n\beta c_1\|o_A\|  |\partial (A\cup B)_r|\exp{[-C_\beta d(A,C)/(2r)]}=O(e^{-d(A,C)}).
\end{align}
Here, as previously defined in Eqs.~(\ref{Seq:defdeltaH}) and (\ref{Seq:defdeltaHab}), we define $\Delta H$ by 
    $\Delta H = -H^{(BC)} +\Delta H^{(A+B)}$,
where $H^{(BC)}$ is the interaction Hamiltonian defined as the terms in original Hamiltonian $H$ that are  supported on both of the region $B$ and $C$,
and $\Delta H^{(A+B)}$ is defined by 
    $\Delta H^{(A+B)}= -\frac{1}{\beta}\log\tr_C[\rho_{\beta}(H)] - H^{(A+B)}$,
where $H^{(A+B)}$ is the Hamiltonian defined as the terms in original Hamiltonian $H$ that are  supported on either of the region $A$ and $B$.
\end{theorem}
\begin{proof}

Assumption \ref{assumption1} means that the support of the deviation $\Delta H$ is exponentially localized around the boundary of region $B$.
We can choose the parameter $l$ and here we fix $l=d(A,C)/2$ for simplicity, which leads to $O(e^{-l})=O(e^{-d(A,C)})$.
In order to bound the accuracy of the LVP, we derive the following lemma that is useful to discuss the deviation between Gibbs states:
\begin{lemma}\label{lemma1}
If the two Hamiltonian $H_1$ and $H_2$ satisfy the condition $\|H_1-H_2\|\le \varepsilon$, then
\begin{align}
\|\rho_\beta(H_1)-\rho_\beta(H_2)\|&\le 2\beta\|H_1-H_2\| \le  2\beta\varepsilon.
\end{align}
Furthermore, a similar bound holds:
\begin{align}
\|\rho_\beta(H_1)-\rho_\beta(H_2)\|_1&\le 2\beta\|H_1-H_2\| \le  2\beta\varepsilon,
\end{align}
where $\|\bullet\|_1$ denotes the trace distance.
\end{lemma}
\begin{proof}
We first define $\rho_\lambda = e^{-\beta (H_1+\lambda (H_2-H_1))}/\tr[e^{-\beta (H_1+\lambda (H_2-H_1))}]$ and then we obtain
\begin{align}
    \|\rho_\beta(H_1)-\rho_\beta(H_2)\|
    &=\left\|\beta\int_0^1 d\lambda\int_0^1 d\tau \rho_\lambda^\tau (H_1-H_2)\rho_\lambda^{1-\tau} - \rho_\lambda\mathrm{Tr}[\rho_\lambda^\tau (H_1-H_2)\rho_\lambda^{1-\tau}] \right\|\\
    &\le \beta\int_0^1 d\lambda\int_0^1 d\tau \left\|\rho_\lambda^\tau (H_1-H_2)\rho_\lambda^{1-\tau}\right\| +\left\| \rho_\lambda\mathrm{Tr}[\rho_\lambda (H_1-H_2)] \right\|\\
    &\le \beta\int_0^1 d\lambda\int_0^1 d\tau \|\rho_\lambda^\tau\|\|H_1-H_2\|\|\rho_\lambda^{1-\tau}\| +\|H_1-H_2\|\\
    &\le 2\beta\|H_1-H_2\| \le  2\beta\varepsilon,
\end{align}
where we employ the relations $\|PQ\|\ \le \|P\|\|Q\|$ and $|\tr[\rho P]|\le\|\rho \|_1\|P\| = \|P\|$ for arbitrary operators $P$, $Q$, and $\rho$, satisfying $\tr[\rho]=1$ and $\rho\ge 0$. The latter inequality can be proven in a similar manner, using the H\"{o}lder's inequality: $\|PQ\|_1\le \|P\|_1\|Q\|$.
\end{proof}
From Assumption \ref{assumption1}, we can replace $\Delta H^{(A+B)}$ with $\Delta H^{(A+B)}_{\rm local}$ with an exponentially small error regarding $d(A,C)$.
The following inequality simplifies the generalized correlation function.
\begin{proposition} \label{proposision2}
(Theorem 13 in Ref.~\cite{PhysRevX.12.021022})
For an arbitrary observable $O_A$ ($O_C$) supported on the region $A$ ($C$), the following inequality holds
\begin{align}
    |\mathrm{Tr}[\rho^\tau O_A \rho^{1-\tau}O_C]-\mathrm{Tr}[\rho O_A O_C]|&\le C_\beta' \min{(|\partial A|,|\partial C|)}\|O_A\|\|O_C\|e^{-d(A,C)/\xi_\beta'}=O(e^{-d(A,C)}), 
\end{align}
for an arbitrary parameter $0\le \tau \le 1$ and for any Gibbs state $\rho$ of an arbitrary short-range Hamiltonian with any temperature.
The parameters $C_\beta'$ and $\xi_\beta'$ depend only on the temperature $\beta$ and other system parameters such as $g$ and does not depend on the system size $N$. $|\partial A|$ and $|\partial C|$ are the number of sites on the boundary region.
\end{proposition}
Note that the Hamiltonian $H+\Delta H$ can be regarded as a short-range Hamiltonian if we replace $\Delta H^{(A+B)}$  with $\Delta H^{(A+B)}_{\rm local}$, and we define $H+\Delta H_{\rm local}$ by replacing $\Delta H^{(A+B)}$  with $\Delta H^{(A+B)}_{\rm local}$. 
We also define $\rho_\lambda^{\rm (local)} =e^{-n\beta (H+\lambda \Delta H_{\rm local})}/\tr[e^{-n\beta (H+\lambda \Delta H_{\rm local})}]$.
Using these replacements and by combining Assumption \ref{assumption1}, Proposition \ref{proposision2}, Lemma \ref{lemma1}, and the assumption of the exponential clustering of the two-point correlation function, we obtain
\begin{align}
    |D^{(n)}(o_A)|&=\left|\frac{\mathrm{Tr}[e^{-n\beta(H+\Delta H)} o_A]}{\mathrm{Tr}[e^{-n\beta(H+\Delta H)}]}-\frac{\mathrm{Tr}[e^{-n\beta H} o_A]}{\mathrm{Tr}[e^{-n\beta H}]}\right|\\
    &\le \left|\frac{\mathrm{Tr}[e^{-n\beta(H+\Delta H_{\rm local})} o_A]}{\mathrm{Tr}[e^{-n\beta(H+\Delta H_{\rm local})}]}-\frac{\mathrm{Tr}[e^{-n\beta H} o_A]}{\mathrm{Tr}[e^{-n\beta H}]}\right|+\left|\frac{\mathrm{Tr}[e^{-n\beta(H+\Delta H_{\rm local})} o_A]}{\mathrm{Tr}[e^{-n\beta(H+\Delta H_{\rm local})}]}-\frac{\mathrm{Tr}[e^{-n\beta(H+\Delta H)} o_A]}{\mathrm{Tr}[e^{-n\beta(H+\Delta H)}]}\right|\\
    &\le \left|\frac{\mathrm{Tr}[e^{-n\beta(H+\Delta H_{\rm local})} o_A]}{\mathrm{Tr}[e^{-n\beta(H+\Delta H_{\rm local})}]}-\frac{\mathrm{Tr}[e^{-n\beta H} o_A]}{\mathrm{Tr}[e^{-n\beta H}]}\right|+2n\beta\|o_A\|\|\Delta H^{(A+B)}-\Delta H^{(A+B)}_{\rm local}\| \\
    &\le n\beta\int_0^1 d\lambda\int_0^1 d\tau ~\left|\mathrm{Corr}_{\rho^{\rm (local)}_\lambda}^\tau (o_A,\Delta H_{\rm local})\right|+2n\beta\|o_A\|\|\Delta H^{(A+B)}-\Delta H^{(A+B)}_{\rm local}\|\\
    &\le n\beta\int_0^1 d\lambda\int_0^1 d\tau ~\left|\mathrm{Tr}[\rho^{\rm (local)}_\lambda o_A\Delta H_{\rm local}]-\mathrm{Tr}[\rho^{\rm (local)}_\lambda o_A]\mathrm{Tr}[\rho^{\rm (local)}_\lambda \Delta H_{\rm local}]\right|\notag\\
    &\qquad+C_\beta' \min{(|\partial A|,|\partial C|)}\|o_A\|\|\Delta H_{\rm local}\|e^{-d(A,C)/\xi_\beta'}+2n\beta\|o_A\|\|\Delta H^{(A+B)}-\Delta H^{(A+B)}_{\rm local}\|\\
    &\le n\beta\left|\mathrm{Tr}[\rho_{\lambda_{\rm max}}^{\rm (local)} o_A\Delta H_{\rm local}]-\mathrm{Tr}[\rho_{\lambda_{\rm max}}^{\rm (local)} o_A]\mathrm{Tr}[\rho_{\lambda_{\rm max}}^{\rm (local)} \Delta H_{\rm local}]\right|\notag\\
    &\qquad+C_\beta' \min{(|\partial A|,|\partial C|)}\|o_A\|\|\Delta H_{\rm local}\|e^{-d(A,C)/\xi_\beta'}+2n\beta\|o_A\|\|\Delta H^{(A+B)}-\Delta H^{(A+B)}_{\rm local}\|\\
    &\le c'n\beta\|o_A\|\|\Delta H_{\rm local}\|e^{-d(A,C)/\xi'}+C_\beta' \min{(|\partial A|,|\partial C|)}\|o_A\|\|\Delta H_{\rm local}\|e^{-d(A,C)/\xi_\beta'}\notag\\
    &\qquad  +2n\beta c_1\|o_A\|  |\partial (A\cup B)_r|\exp{[-C_\beta d(A,C)/(2r)]}  
    \label{Seq:proofoftheorem1}
\end{align}
where $\lambda_{\rm max}={\rm arg~max}_{0\le \lambda \le 1} \left|\mathrm{Tr}[\rho_{\lambda}^{\rm (local)} o_A\Delta H_{\rm local}]-\mathrm{Tr}[\rho_{\lambda}^{\rm (local)} o_A]\mathrm{Tr}[\rho_{\lambda}^{\rm (local)} \Delta H_{\rm local}]\right|$ is the  value which maximizes the two-point correlation function. 
It completes our proof of Theorem \ref{Theorem3} by noting that Eq.~(\ref{Seq:proofoftheorem1}) implies
\begin{align}
    |D^{(n)}(o_A)|=O(e^{-d(A,C)}).
\end{align}
\end{proof}


\section{Exponential clustering from area law of entanglement entropy}
In this section, we show that $\frac{\|\rho^{n-1}_C\|}{\mathrm{Tr}[(\rho_{C})^n]}=O(1)$ holds when the entanglement entropy $S(\rho_C):=-\tr[\rho_C \log{\rho_C}]$ is a constant, which holds for instance in one-dimensional local Hamiltonian with constant spectral gap~\cite{Hastings_2007}.
Assuming $S(\rho_C) = O(1)$, we obtain
\begin{align}
\frac{\|\rho^{n-1}_C\|}{\mathrm{Tr}[(\rho_{C})^n]}\le \frac{\|\rho^{n-1}_C\|}{\|\rho^{n}_C\|}=\frac{1}{\lambda^{\rm max}_C},    
\end{align}
where $\lambda^{\rm max}_C$ is the largest eigenvalue of $\rho_C$.
Moreover, by using the inequality $-\log{\lambda^{\rm max}_C}\le S(\rho_C)$ which follows from the fact that the Shannon entropy is larger than the min-entropy, we obtain 
\begin{align}
    \frac{1}{\lambda^{\rm max}_C}\le \exp{[S(\rho_C)]}=O(1) 
\end{align}


\section{Additional information of numerical results in the main text}
In this section, we describe the additional information of the numerical results in Figs. 2, 3, and 4 in the main text.
\subsection{Cooling of canonical Gibbs states}
We provide details on $\chi_{\rm LVP},\chi_{\rm FVP}$ of the total energy, as shown in Fig.~2 in the main text.
The MSE $\chi$ can be decomposed into two elements 
$\chi = ({\rm Bias})^2 + \frac{{\rm Var}}{N_{\rm shot}}$:
the bias from the true value, i.e.~the expectation values of FVP, and the variance $\frac{{\rm Var}}{N_{\rm shot}}$ resulting from the number of measurements $N_{\rm shot}$.
The explicit forms of MSEs are given by
\begin{align}
    \chi_{\rm FVP}&=\sum_i\frac{{\rm Var}(\hat{h}_i)_{{\rm FVP}}}{N_{\rm shot}}  \\
    \chi_{\rm LVP}&=\sum_i|D^{(n)}(\hat{h}_i)|^2+\frac{{\rm Var}(\hat{h}_i)_{{\rm LVP}}}{N_{\rm shot}},  
\end{align}
where ${\rm Var}(\hat{h}_i)_{{\rm FVP}}$ and ${\rm Var}(\hat{h}_i)_{{\rm LVP}}$ are defined in  Eq.~(\ref{appeq:variance_FVP}) and Eq.~(\ref{appeq:variance_LVP}), respectively. 

\subsection{Error mitigation of ground state under local depolarizing noise}
\begin{figure*}[htbp]
\begin{minipage}[b]{0.41\linewidth}
    \centering
  \includegraphics[width=1.0\textwidth]{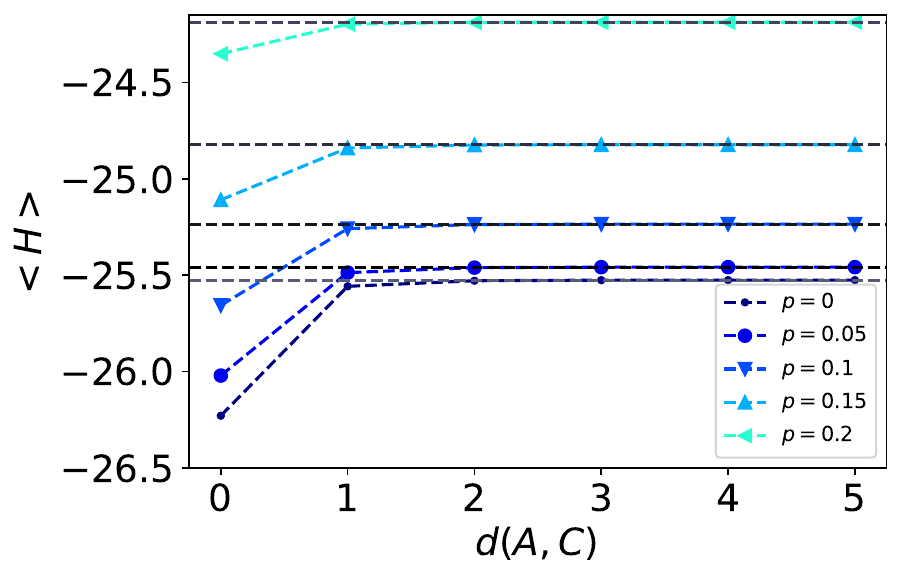}
  \subcaption{Expectation values}
\end{minipage}
  \hspace{0.04\columnwidth}
\begin{minipage}[b]{0.41\linewidth}
    \centering
  \includegraphics[width=1.0\textwidth]{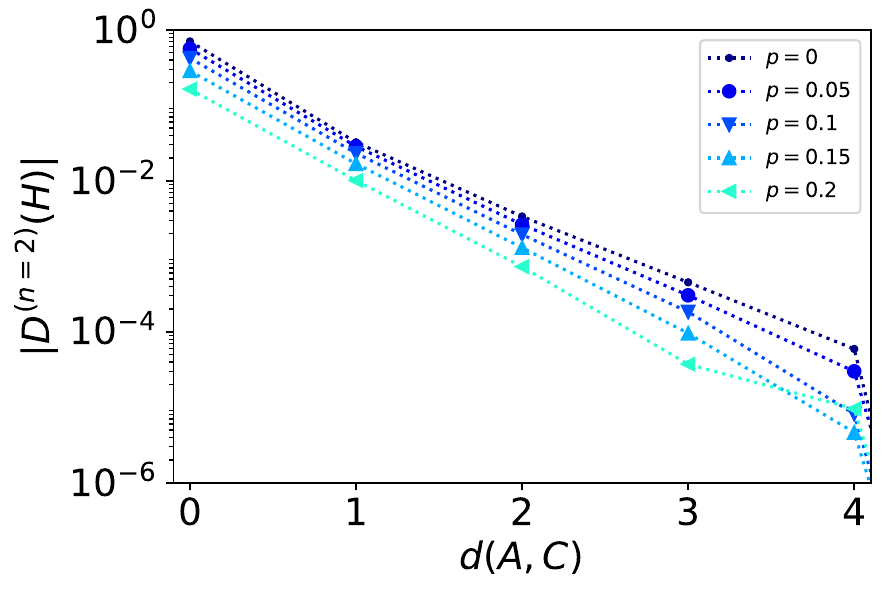}
  \subcaption{Total deviation of expectation values}
  \end{minipage}
  \caption{Expectation values of $H_{\rm TFI}$ for the ground state of one-dimensional transverse field Ising model on a non-critical point with local depolarizing noise. We fix $n=2$, $N=12$, and $p=0,0.05,0.1,0.15,0.2$. (a) Expectation values of LVP against $d(A,C)$. Horizontal dashed lines indicate the expectation values of FVP. (b) The difference of the expectation values between FVP and LVP against $d(A,C)$. 
  }
  \label{Fig:expectationvalue}
\end{figure*}
Figure \ref{Fig:expectationvalue} (a) shows the expectation values of the ground state energy of the Hamiltonian in Eq.~(\ref{eq:transverseising}) on a non-critical point under single-qubit local depolarizing noise.
As shown in Fig.~\ref{Fig:expectationvalue} (a), we can confirm the expectation values of LVP approaches those of FVP as $d(A,C)$ increases. 
Furthermore, the difference between LVP and FVP in Fig.~\ref{Fig:expectationvalue} (b) decreases exponentially, which is consistent with the analytical results for the noiseless case as presented in the main text.
For small $d(A,C)$, the expectation values of LVP are lower than the exact values and thus unphysical. 
However, these results have a meaning of upper bounds of the results obtained by using the variational reduced density matrix method~\cite{nakata2001variational} without any constraints from $N$-representability conditions~\cite{doi:10.1126/science.abb9811} (or called quantum marginal problem~\cite{klyachko2006quantum}).

In addition to the specific expectation values, we calculate the trace distance $\|D^{(n)}\|_1$ between LVP and FVP, which is also referred to as local indistinguishability~\cite{brandao2019finite},
\begin{align}
    \|D^{(n)}\|_1&=\left\|\frac{\mathrm{Tr}_B[(\rho^{(A+B)})^n ]}{\mathrm{Tr}_{A+B}[(\rho^{(A+B)})^n]}-\frac{\mathrm{Tr}_{B+C}[\rho^n]}{\mathrm{Tr}[\rho^n]}\right\|_1\notag \\
    &=\left\|\rho_{\rm LVP}^{(A)}-\rho_{\rm FVP}^{(A)} \right\|_1,
    \label{Eq:tracedistance}
\end{align}
where $\left\|D^{(n)}\right\|_1$ denotes the trace distance between two operators of $\rho_{\rm LVP}^{(A)}\left(=\frac{\mathrm{Tr}_B[(\rho^{(A+B)})^n ]}{\mathrm{Tr}_{A+B}[(\rho^{(A+B)})^n]}\right)$ and $\rho_{\rm FVP}^{(A)}\left(=\frac{\mathrm{Tr}_{B+C}[\rho^n]}{\mathrm{Tr}[\rho^n]}\right)$.
The difference of the expectation values of any observable $O$ can be bounded by the trace distance:
$|\mathrm{Tr}[\rho_1 O]-\mathrm{Tr}[\rho_2 O]|\le 2\|O\|_{\rm max}\|\rho_1 - \rho_2\|_1$~\cite{PhysRevX.11.041036,Koczor_2021},
where $\|O\|_{\rm max}$ is the operator norm of the observable $O$.
\begin{figure}[htbp]
    \centering
  \includegraphics[width=0.6\textwidth]{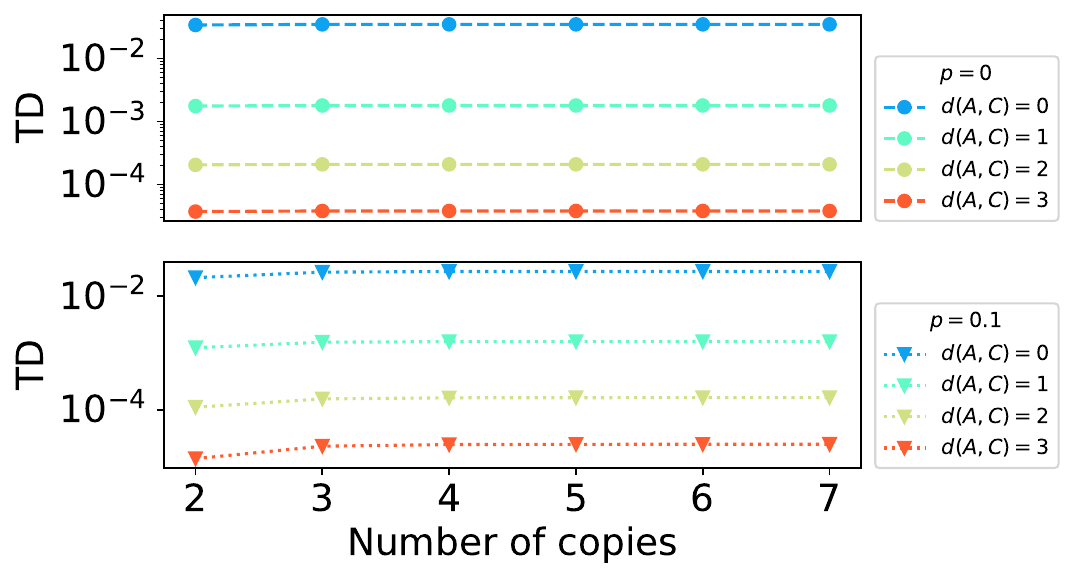}
  \caption{Trace distance (TD) in Eq.~(\ref{Eq:tracedistance}) between $\rho_{\rm LVP}^{(A)}$ and $\rho_{\rm FVP}^{(A)}$ against the number of copies $n$. Here, we consider error mitigation for the noisy ground state of one-dimensional transverse field Ising model with $\lambda=2$ (non-critical) under single-qubit local depolarizing noise model. The system size is taken as $N=12$.
  }
  \label{Fig:TDIsingnoncritical}
\end{figure}
In Fig.~\ref{Fig:TDIsingnoncritical}, we show the $\|D^{(n)}\|_1$ for the noisy ground state of one-dimensional transverse field Ising model with $\lambda=2$ (non-critical). 
We see that the $\|D^{(n)}\|_1$ rapidly converges against $n$ (we numerically confirm the exponential convergence against $n$ for this case) so that $n=2$ or $n=3$ is sufficient for practical situations.

\subsection{Noisy Gibbs states}
In the main text, we have shown that the estimation error can be indeed suppressed even when we combine the two usages, namely the error mitigation and cooling. Here, we provide numerical results on the measurement cost to show the advantage of utilizing the LVP.
Let us consider Gibbs states of the TFI Hamiltonian with $\lambda=2$ as well as in the main text.
\begin{figure}[htbp]
    \centering
  \includegraphics[width=0.6\textwidth]{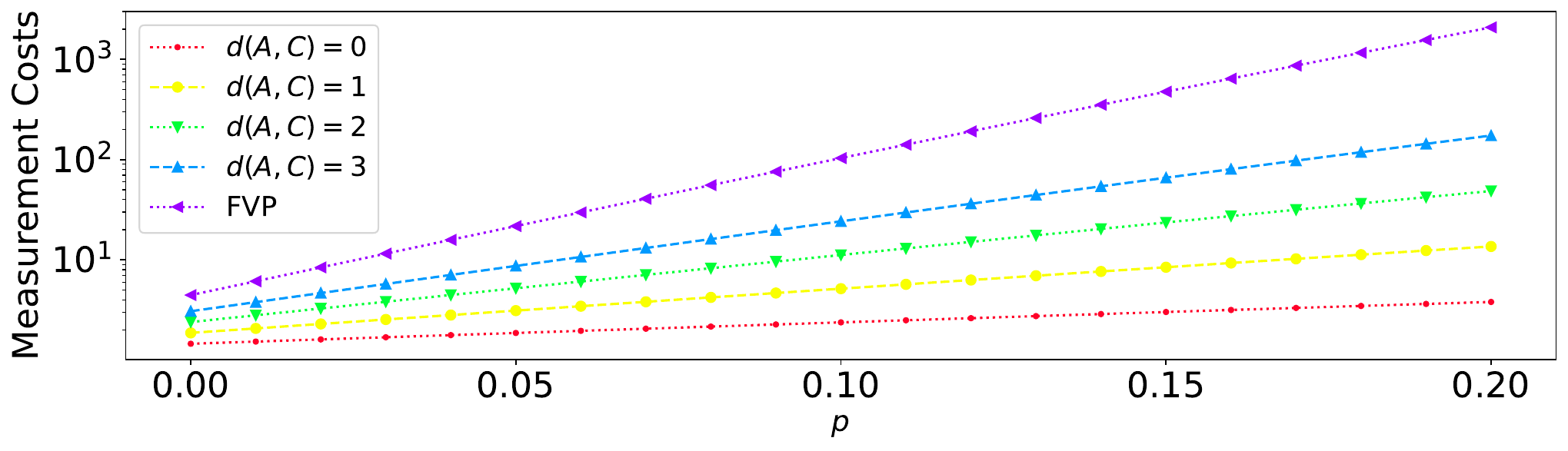}
  \caption{Measurement costs of FVP and LVP for cooling and error mitigation with regard to various error rate $p$.
  }
  \label{Fig:VCVDIsingnoncriticalnoisy}
\end{figure}
In Fig.~\ref{Fig:VCVDIsingnoncriticalnoisy}, we show the measurement costs of FVP and LVP for both cooling and error mitigation under local depolarizing noise.
We can see that the growth of the measurement costs for LVP is slower than that of FVP, which is the same behavior as the case of the noisy ground states.

\section{Other examples of the decay behavior}
In this section, we examine the $d(A,C)$ dependence of the deviation for several examples.
\subsection{One-dimensional critical system under error mitigation}\label{supplesec:critical}
We investigate how the LVP performs in the context of error mitigation when the  condition of the exponential clustering in Eq.~(\ref{Eq:exponentialclustering}) is not satisfied.
As an example, we consider the ground state of $H_{\rm TFI}$ in Eq.~(\ref{eq:transverseising}) at $\lambda=1$ (critical), which shows power-law decay of two-point correlation functions and thus does not obey the exponential decay of Eq.~(\ref{Eq:exponentialclustering}).
Nonetheless, the effect of noise can be efficiently suppressed. 
Figs.~\ref{Fig:expectationvalues_critical} (a) and (b) plot the mitigated expectation values and their deviation between FVP and LVP.
The deviation does not decay exponentially, but it still converges to zero.
In particular, for a noiseless case, it is expected that the deviation shows the power-law decay (See Fig.~\ref{Fig:HSdistanceinfty} in the following section, which contains the calculation for larger sizes based on fermionic Gaussian states).

\begin{figure*}[htbp]
\begin{minipage}[b]{0.45\linewidth}
  \includegraphics[width=1.0\textwidth]{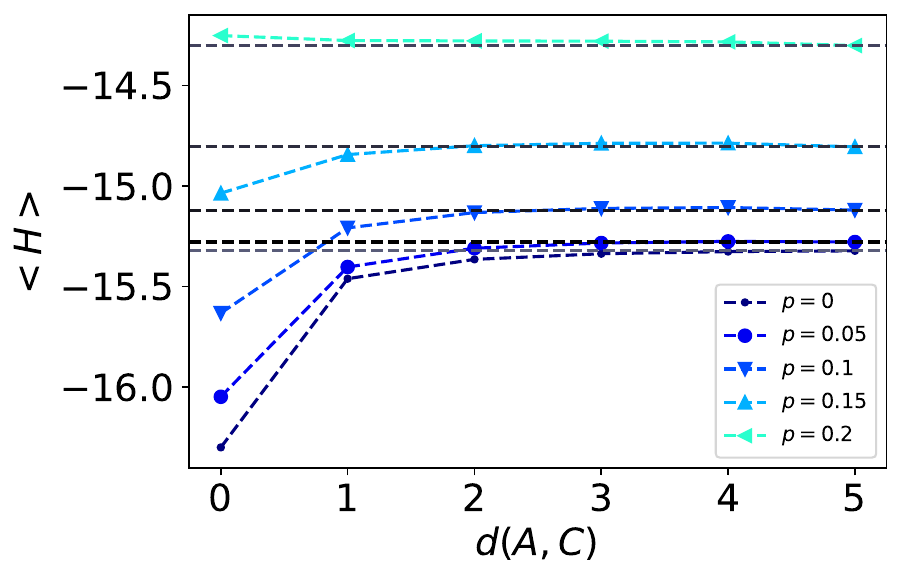}
  \subcaption{Expectation values of energy}
\end{minipage}
  \hspace{0.04\columnwidth}
\begin{minipage}[b]{0.45\linewidth}
    \centering
  \includegraphics[width=1.0\textwidth]{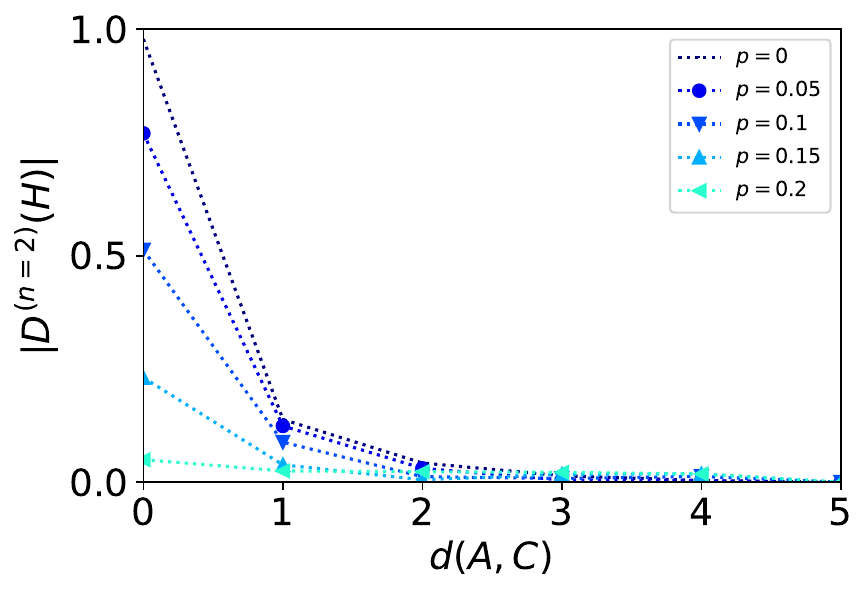}
  \subcaption{Deviation in the expectation values of energy}
  \end{minipage}
  \caption{Expectation values of the Hamiltonian estimated by the LVP for the noisy ground state of one-dimensional transverse field Ising model on a critical point with local depolarizing noise. We fix $n=2$, $N=12$, and $p=0,0.05,0.1,0.15,0.2$. (a) Expectation values against $d(A,C)$. (b) The difference of the expectation values between FVP and LVP against $d(A,C)$.
  }
  \label{Fig:expectationvalues_critical}
\end{figure*}

\begin{figure}[htbp]
    \centering
  \includegraphics[width=0.6\textwidth]{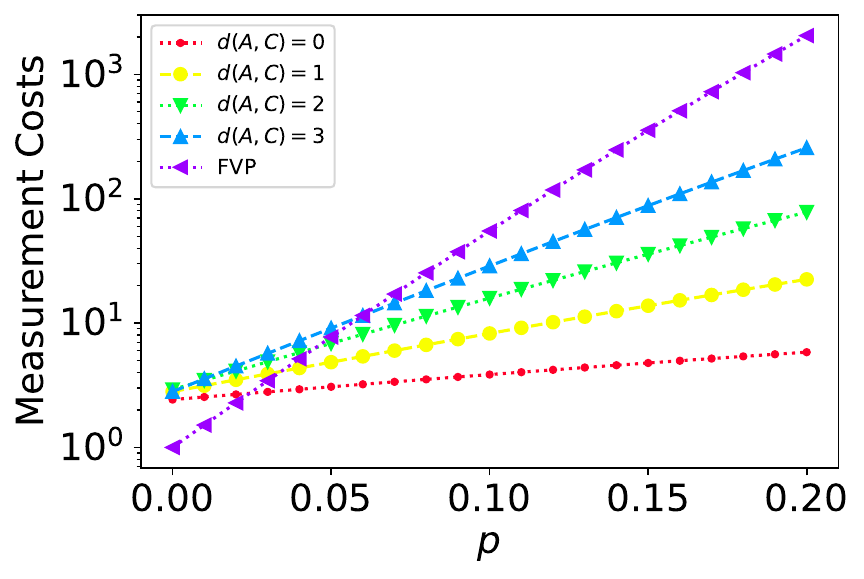}
  \caption{Measurement costs against error rate $p$ to perform the LVP and FVP for the noisy ground state of one-dimensional transverse field Ising model with $\lambda=1$ (critical) under local depolarizing noise. 
  The total number of sites is $N=12$.
  }
  \label{Fig:cost_critical}
\end{figure}
Figure \ref{Fig:cost_critical} shows the measurement costs to perform purification-based quantum simulation for the noisy ground state of the Hamiltonian in Eq.~(\ref{eq:transverseising}) at $\lambda=1$ (critical) under local depolarizing noise.
We can see that the growth of the measurement costs for LVP is slower than that of FVP, which is almost the same behavior as the non-critical case in the main text.

\begin{figure}[htbp]
    \centering
  \includegraphics[width=0.6\textwidth]{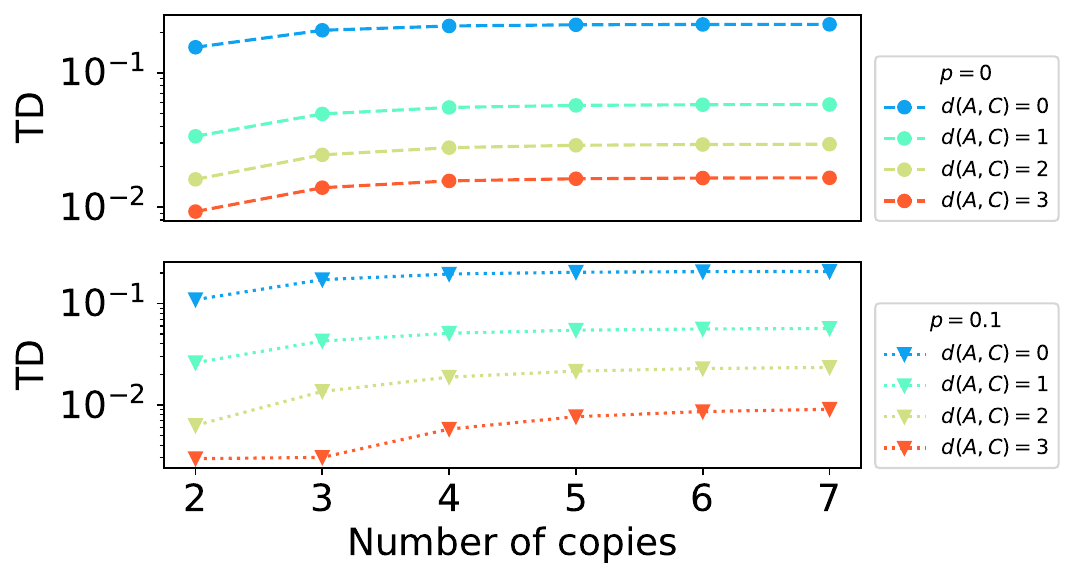}
  \caption{Trace distance (TD) in Eq.~(\ref{Eq:tracedistance}) between $\rho_{\rm LVP}^{(A)}$ and $\rho_{\rm FVP}^{(A)}$ against the number of copies $n$ for the noisy ground state of one-dimensional transverse field Ising model on a critical point under local depolarizing noise model. The total number of sites is $N=12$.
  }
  \label{Fig:TDIsingcritical}
\end{figure}
In Fig.~\ref{Fig:TDIsingcritical}, we also show the $\|D^{(n)}\|_1$ for the ground state of the one-dimensional transverse field Ising model in Eq.~(\ref{Eq:tracedistance}) with different copies $n$ (critical case). 
The convergence against $n$ shows the same behavior as the non-critical case.

\subsection{Two dimensional systems for zero temperature}
\subsubsection{Model}
We consider a non-interacting fermionic system:
\begin{align}
    H=-t\sum_{\braket{\bm{r},\bm{r}'}} (c_{\bm{r}}^\dagger c_{\bm{r}'}+c^\dagger_{\bm{r}'}c_{\bm{r}}) -\mu \sum_{\bm{r}} c_{\bm{r}}^\dagger c_{\bm{r}}.
    \label{eq:twodimensionalfreefermion}
\end{align}
Here, $\bm{r}=(x,y)$ denotes a two dimensional coordinate, $c_{\bm{r}}^\dagger, c_{\bm{r}}$ are creation and annihilation operators at site $\bm{r}$, $\sum_{\braket{\bm{r},\bm{r}'}}$ denotes a sum over a nearest-neighbor site on a square lattice, and $t(>0)$ and $\mu$ denote a hopping integral and a chemical potential, respectively. 
This Hamiltonian can be diagonalized as
\begin{align}
    H=\int_{-\pi}^{\pi}dk_x\int_{-\pi}^{\pi}dk_y (\varepsilon_{k_x,k_y}-\mu) c_{k_x,k_y}^\dagger c_{k_x,k_y},
\end{align}
where we define
\begin{align}
    \varepsilon_{k_x,k_y} &=-2t(\cos{k_x}+\cos{k_y})\\
    c_{k_x,k_y} &= \frac{1}{2\pi}\sum_{x,y} e^{ik_xx+ik_yy}c_{x,y}.
\end{align}
Throughout the following analysis, we assume that the total system is infinite for eliminating the effect of the boundary, but we can similarly conduct a calculation in a finite system by replacing the integral $\int dk_xdk_y$ into a sum $\sum_{k_x,k_y}$ over some discrete set determined by $N$.
The ground state of the Hamiltonian can be characterized by the expectation values:
\begin{equation}
    \braket{c_{k_x,k_y}^\dagger c_{k_x',k_y'}}=\delta(k_x-k'_x)\delta(k_y-k'_y)\times \Theta(\varepsilon_{k_x,k_y}\le \mu),
\end{equation}
where $\delta(\bullet)$ is the Dirac delta function and $\Theta(\bullet)$ is the Heaviside step function.

The expectation values of the product of two creation and annihilation operators are obtained as
\begin{align}
    &\braket{c_{\bm{r}}^\dagger c_{\bm{r}'}}\\
    &=\frac{1}{(2\pi)^2}\int_{-\pi}^{\pi}dk_x\int_{-\pi}^{\pi}dk_y\int_{-\pi}^{\pi}dk_x'\int_{-\pi}^{\pi}dk_y'e^{-ik_xx-ik_yy+ik_x'x'+ik_y'y}\delta(k_x-k_x')\delta(k_y-k_y')\times \Theta(\varepsilon_{k_x,k_y}\le \mu)\\
    &=\frac{1}{(2\pi)^2}\int_{-\pi}^{\pi}dk_x\int_{-\pi}^{\pi}dk_ye^{ik_x(x'-x)+ik_y(y'-y)}\times \Theta(\varepsilon_{k_x,k_y}\le \mu)\\
    &=\frac{1}{\pi^2}\int_{0}^{k_F}dk_x\cos{k_x(x'-x)} \left[ \delta_{y,y'}\arccos{(-\mu/2-\cos{k_x})}+(1-\delta_{y,y'})\frac{\sin{[(\arccos{(-\mu/2-\cos{k_x})})(y'-y)]}}{y'-y} \right],
\end{align}
where $k_F$ is a Fermi wavenumber determined by $\varepsilon_{k_{F},0}=\varepsilon_{0,k_{F}}=\mu$.
Although we can take the value of $\mu$ as $-2t\le\mu\le 2t$, we assume $-2t\le\mu\le 0$ without loss of generality.
Here we define a matrix $\Lambda$ as
\begin{align}
    \Lambda_{\bm{r},\bm{r}'}:=\braket{c_{\bm{r}}^\dagger c_{\bm{r}'}}\qquad (\bm{r},\bm{r}'\in {A}).
\end{align}
This matrix is Hermitian and can be diagonalized by using a unitary operator $U$ as
\begin{align}
    U^\dagger \Lambda U = \Xi,
\end{align}
where $\Xi$ is a diagonal matrix and we define the diagonal elements as $\xi_b$. We also define a new annihilation operator 
\begin{align}
    c_b = \sum_{{\bm{r}}\in  A} U_{{\bm{r}},b}c_{\bm{r}}
\end{align}
in order to simultaneously diagonalize the matrix $\Lambda$ and a reduced density matrix in region $A$.
Using the annihilation operator $c_b$, we obtain
\begin{align}
    \mathrm{Tr}[\rho_{A} c_b^\dagger c_{b'}]=\xi_{b} \delta_{b,b'},
\end{align}
and $\rho_{A}$ can be written as
\begin{align}
    \rho_{A}&=\frac{1}{Z_A}\exp{\left[-\sum_{b} \zeta_b c_b^\dagger c_b\right]}\\
    Z_A&= \tr\left[\exp{\left[-\sum_{b} \zeta_b c_b^\dagger c_b\right]}\right]\\
    \zeta_b&= \log{(\xi_{b}^{-1}-1)}.
\end{align}
$b$ can be regarded as a good quantum number and that is the reason for introducing the annihilation operator $c_b$.
\subsubsection{Calculation of LVP}
Since we are considering a free-fermionic system, the information regarding the reduced density operator as well as expectation values are all encoded in the correlation matrix $\Lambda$~\cite{Peschel_2003,Peschel_2009}..
Moreover, we find that the density operators involved in the LVP are also deeply related with $\Lambda$.
To see this, let us introduce the following:
\begin{align}
    (\rho_{A})^2&=\frac{1}{(Z_{A})^2 } \exp{\left[-\sum_{b} 2\zeta_b c_b^\dagger c_b\right]},
\end{align}
The expectation values of LVP are obtained by
\begin{align}
    \frac{\mathrm{Tr}[(\rho_{A})^2 c_{\bm{r}}^\dagger c_{\bm{r}'}]}{\mathrm{Tr}[(\rho_{A})^2]}&=\sum_{b,b'}U_{{\bm{r}},b}U^*_{{\bm{r}}',b'}\frac{\mathrm{Tr}[(\rho_{A})^2 c_b^\dagger c_{b'}]}{\mathrm{Tr}[(\rho_{A})^2]}\\
    &=\sum_{b,b'}U_{{\bm{r}},b}U^*_{{\bm{r}}',b'}\frac{\delta_{b,b'}}{1+e^{2\zeta_b}}\\
    &=\sum_{b}U_{{\bm{r}},b}U^*_{{\bm{r}}',b}\frac{1}{1+\left(\frac{1-\xi_b}{\xi_b}\right)^2 }\\
    &=[\Lambda^2 (\Lambda^2+(1-\Lambda)^2)^{-1} ]_{{\bm{r}},{\bm{r}}'}.
\end{align}
Similarly, the expectation values for $n$-copy cases are given by
\begin{align}
    \frac{\mathrm{Tr}[(\rho_{A})^n c_{\bm{r}}^\dagger c_{\bm{r}'}]}{\mathrm{Tr}[(\rho_{A})^n]}&=[\Lambda^n (\Lambda^n+(1-\Lambda)^n)^{-1} ]_{{\bm{r}},{\bm{r}}'}.
\end{align}
\subsubsection{Numerical results of trace distance}
For a simplest case, we consider $N_A=1$.
In this case, the reduced density matrix can be determined by a single expectation value
\begin{align}
    \braket{c_{\bm{r}}^\dagger c_{\bm{r}}} = \mathrm{Tr}[\rho_{A} c_{\bm{r}}^\dagger c_{\bm{r}}]\qquad ({\bm{r}}\in {A})
\end{align}
The trace distance between two states $\rho_{A}^{(1)}$ and $\rho_{A}^{(2)}$ can be expressed as the difference of the local density $\braket{c_{\bm{r}}^\dagger c_{\bm{r}}}$:
\begin{figure}[htbp]
    \centering
  \includegraphics[width=0.6\textwidth]{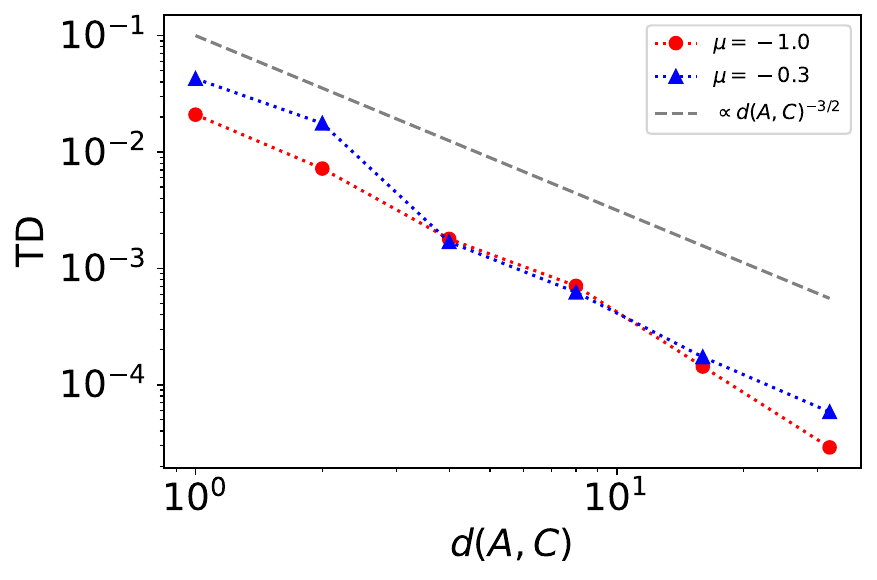}
  \caption{Trace distance in Eq.~(\ref{Eq:twodimensionalTD}) between $\rho_{A}^{(1)}$ (LVP) and $\rho_{A}^{(2)}$ (FVP) against $d(A,C)$ for the ground state of a two-dimensional non-interacting fermionic system. We fix $N_A=1$, $n=2$ copies, and $N_C\rightarrow \infty$.
  }
  \label{Fig:twodimensionalfreefermion}
\end{figure}
\begin{align}
    \|\rho_{A}^{(1)}-\rho_{A}^{(2)}\|_1 &=\frac{1}{2}\mathrm{Tr}|\rho_{A}^{(1)}-\rho_{A}^{(2)}| \\
    &= |m_1-m_2|,
    \label{Eq:twodimensionalTD}
\end{align}
where we define $m_n = \mathrm{Tr}[\rho_{A}^{(n)} c_{\bm{r}}^\dagger c_{\bm{r}}]$ ($n=1,2$).
Note that the Hilbert-Schmidt distance is the same value except for the overall factor $\sqrt{2}$
\begin{align}
    \|\rho_{A}^{(1)}-\rho_{A}^{(2)}\|_2 &=\sqrt{\mathrm{Tr}[(\rho_{A}^{(1)}-\rho_{A}^{(2)})^2]} \\
    &= \sqrt{2}|m_1-m_2|.
\end{align}
Figure~\ref{Fig:twodimensionalfreefermion} plots the trace distance between $\rho_{A}^{(1)}$ (LVP) and $\rho_{A}^{(2)}$ (FVP) for a two-dimensional system.
This graph shows the power-law decay of trace distance and its power is almost equal to $3/2$ for all $\mu$.

\subsection{Two dimensional systems for finite temperature}
We consider a non-interacting fermionic system as already defined in Eq.~(\ref{eq:twodimensionalfreefermion}).
The expectation values of the product of two creation and annihilation operators are obtained
\begin{align}
    \braket{c_{\bm{r}}^\dagger c_{\bm{r}'}}
    &=\frac{1}{(2\pi)^2}\int_{-\pi}^{\pi}dk_x\int_{-\pi}^{\pi}dk_y\int_{-\pi}^{\pi}dk_x'\int_{-\pi}^{\pi}dk_y'e^{-ik_xx-ik_yy+ik_x'x'+ik_y'y}\delta(k_x-k_x')\delta(k_y-k_y')\times F_\beta(\varepsilon_{k_x,k_y}- \mu)\\
    &=\frac{1}{(2\pi)^2}\int_{-\pi}^{\pi}dk_x\int_{-\pi}^{\pi}dk_ye^{ik_x(x'-x)+ik_y(y'-y)}\times F_\beta(\varepsilon_{k_x,k_y}- \mu)\\
    &=\frac{1}{\pi^2}\int_{0}^{\pi}dk_x\int_{0}^{\pi}dk_y\cos{k_x(x'-x)}\cos{k_y(y'-y)}\times F_\beta(\varepsilon_{k_x,k_y}- \mu),
\end{align}
where $F_\beta(\varepsilon)=1/(1+e^{\beta \varepsilon})$ is the Fermi distribution function at an inverse temperature $\beta$.
Although we can take the value of $\mu$ as $-2t\le\mu\le 2t$, we assume $-2t\le\mu\le 0$ without loss of generality.
Here we define a matrix $\Lambda$ as
\begin{align}
    \Lambda_{\bm{r},\bm{r}'}:=\braket{c_{\bm{r}}^\dagger c_{\bm{r}'}}\qquad (\bm{r},\bm{r}'\in {A}).
\end{align}
This matrix is Hermitian and can be diagonalized by using a unitary operator $U$
\begin{align}
    U^\dagger \Lambda U = \Xi,
\end{align}
where $\Xi$ is a diagonal matrix and we define the diagonal elements as $\xi_b$. We also define a new annihilation operator
\begin{align}
    c_b = \sum_{{\bm{r}}\in A} U_{{\bm{r}},b}c_{\bm{r}}
\end{align}
and we obtain
\begin{align}
    \mathrm{Tr}[\rho_{A} c_b^\dagger c_{b'}]=\xi_{b} \delta_{b,b'}.
\end{align}

\begin{figure}[htbp]
    \centering
  \includegraphics[width=0.6\textwidth]{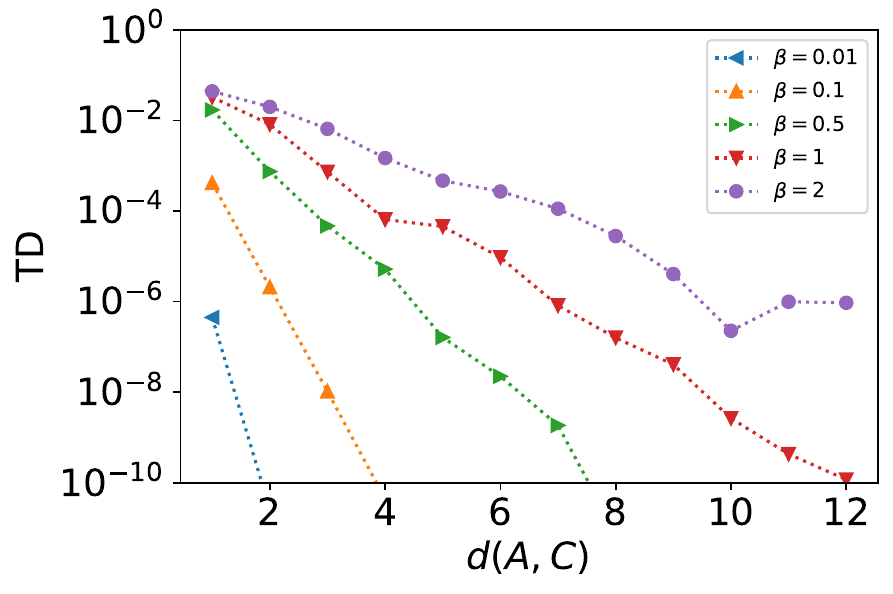}
  \caption{Trace distance in Eq.~(\ref{Eq:twodimensionalTD}) between $\rho_{A}^{(1)}$ (LVP) and $\rho_{A}^{(2)}$ (FVP) against $d(A,C)$ for a two-dimensional non-interacting fermionic system with finite temperatures. We fix $\mu=-0.3$, $N_A=1$, $n=2$ copies, and $N_C\rightarrow \infty$.
  }
  \label{Fig:twodimensionalfreefermion_finite_temperature}
\end{figure}
Figure~\ref{Fig:twodimensionalfreefermion_finite_temperature} plots the trace distance between $\rho_{A}^{(1)}$ (LVP) and $\rho_{A}^{(2)}$ (FVP) for a two-dimensional system with finite temperatures.
This graph shows the exponential decay of trace distance and its correlation length differs in each $\beta$.

\subsection{Three dimensional systems for finite temperature}
We can calculate higher dimensional cases following similar procedures.
In this subsection, we consider the three dimensional non-interacting fermionic system with finite temperatures.
We define a matrix $\Lambda$ as 
\begin{align}
    \Lambda_{\bm{r},\bm{r}'}&:=\braket{c_{\bm{r}}^\dagger c_{\bm{r}'}}\qquad (\bm{r},\bm{r}'\in {A})\\
    &=\frac{1}{\pi^3}\int_{0}^{\pi}dk_x\int_{0}^{\pi}dk_y\int_{0}^{\pi}dk_z\cos{k_x(x'-x)}\cos{k_y(y'-y)}\cos{k_y(z'-z)}\times F_\beta(\varepsilon_{k_x,k_y,k_z}- \mu),
\end{align}
where $\bm{r}=(x,y,z),\bm{r}'=(x',y',z')$ denotes a three dimensional coordinate and
\begin{align}
    \varepsilon_{k_x,k_y,k_z} &=-2t(\cos{k_x}+\cos{k_y}+\cos{k_z}).
\end{align}
\begin{figure}[htbp]
    \centering
  \includegraphics[width=0.6\textwidth]{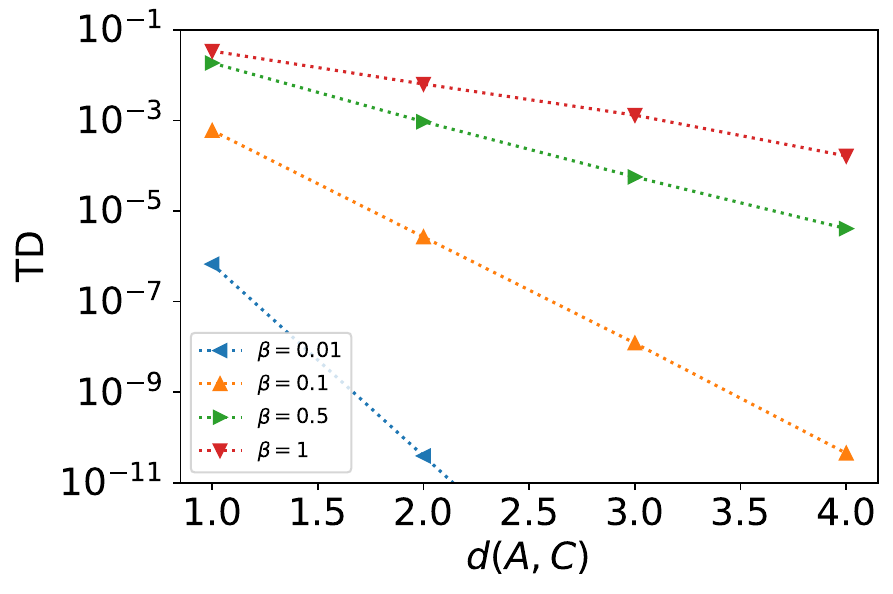}
  \caption{Trace distance in Eq.~(\ref{Eq:twodimensionalTD}) between $\rho_{A}^{(1)}$ (LVP) and $\rho_{A}^{(2)}$ (FVP) against $d(A,C)$ for a three-dimensional non-interacting fermionic system with finite temperatures. We fix $\mu=-0.3$, $N_A=1$, $n=2$ copies, and $N_C\rightarrow \infty$.
  }
  \label{Fig:threedimensionalfreefermion_finite_temperature}
\end{figure}
Figure~\ref{Fig:threedimensionalfreefermion_finite_temperature} plots the trace distance in Eq.~(\ref{Eq:twodimensionalTD}) between $\rho_{A}^{(1)}$ (LVP) and $\rho_{A}^{(2)}$ (FVP) for a three-dimensional system with finite temperatures.
This graph shows the exponential decay of trace distance and its correlation length differs in each $\beta$.

\section{Detail calculations of Hilbert-Schmidt distance between two fermionic Gaussian states}\label{App:criticalHS}
In this section, we consider the the deviation between LVP and FVP for fermionic Gaussian states~\cite{PhysRevLett.90.227902,https://doi.org/10.48550/arxiv.quant-ph/0304098} in order to explain the behavior of the deviation between LVP and FVP for the critical and noiseless case as numerically shown in Sec.~\ref{supplesec:critical}.
Hilbert-Schmidt distance between two fermionic Gaussian states can be calculated in a similar manner of the fidelity~\cite{PhysRevE.89.022102}.
\subsection{Basic properties}
It is convenient to introduce the Majorana representation of fermions as
\begin{align}
    w_{2l-1} &=c_{l}+ c_l^\dagger \\
    w_{2l} &=i(c_{l}- c_l^\dagger),
\end{align}
where $c_{l}$ ($c_{l}^\dagger$) is an annihilation (creation) operator of the $l$-th fermion and
$w_{l}$ denotes the $l$-th Majorana operator that satisfies the canonical anti-commutation relation
\begin{align}
    \{w_{l},w_{l'} \}&=2\delta_{l,l'}.
\end{align}
Fermionic Gaussian states are defined as the Gaussian forms of $w_{l}$.
\begin{align}
    \rho=\frac{1}{Z} \exp{\left[-\frac{i}{4}\sum_{i,j}G_{i,j} w_iw_j\right]},
\end{align}
where $i,j=1,2,\cdots,2N$, $Z=\tr[e^{-\frac{i}{4}\sum_{i,j}G_{i,j} w_iw_j}]$ is a normalization factor, and $G$ is a real and anti-symmetric matrix. 
Fermionic Gaussian states are characterized completely by
 the correlation matrix $M$ as
\begin{align}
    M_{i,j}:=\frac{1}{2}\tr[\rho [w_{i},w_{j} ]].
\end{align}
The relation between the matrix $M$ and $G$ is given by 
\begin{align}
    M&= \tanh{\frac{iG}{2}}.
\end{align}
These Hermitian matrices $M$ and $iG$ can be diagonalized in the same basis.
The purity of fermionic Gaussian state can be calculated as
\begin{align}
    \mathrm{Tr}\left[\rho^2\right]
    &=\sqrt{\mathrm{det}\left[\frac{1+M^2}{2}\right]}.
\end{align}
Hilbert-Schmidt distance between two fermionic Gaussian states can be rewritten as a simplified form
\begin{align}
    \|\rho_1-\rho_2\|_2 &=\sqrt{\mathrm{Tr}[(\rho_1-\rho_2)^2]} \\
    &= \sqrt{\mathrm{Tr}[(\rho_1)^2]+\mathrm{Tr}[(\rho_2)^2]-2\mathrm{Tr}[\rho_1\rho_2]}
\end{align}
The first and second terms are the purity of each fermionic Gaussian states.
The third term can be calculated as
\begin{align}
    \mathrm{Tr}[\rho_1\rho_2]
    &=\sqrt{\mathrm{det}\left[\frac{1+M_2M_1}{2}\right]},
\end{align}
where $M_1$ and $M_2$ are the correlation matrices of $\rho_1$ and $\rho_2$, respectively.
Note that although 
the computation of the Hilbert-Schmidt distance in many-body systems requires exponentially large cost with the number of fermionic modes, this is not the case for fermionic Gaussian states which can be calculated from correlation matrices $M_1$ and $M_2$ so that it requires only polynomial cost.
\subsection{Localized virtual purification for Fermionic Gaussian states}
The LVP procedure preserves the structure of the fermionic Gaussian states and 
doubles all the elements of the matrix $G$ so that
\begin{align}
    \rho_{A+B}^{(n=2)}=\frac{(\rho_{A+B})^2}{\tr[(\rho_{A+B})^2]}=\frac{1}{Z^{(n=2)}_{A+B}} \exp{\left[-\frac{i}{2}\sum_{i,j\in A+B}G_{i,j} w_iw_j\right]},
\end{align}
where $i,j=1,2,\cdots,2(N_A+N_B)$ and $Z^{(n=2)}_{A+B}=\tr[e^{-\frac{i}{2}\sum_{i,j}G_{i,j} w_iw_j}]$ is a normalization factor. 
Therefore the purified states are characterized by the correlation matrix $M_{\rm LVP}^{(n)}$.
Moreover, $M_{\rm LVP}^{(n)}$ can be expressed as a combination of the correlation matrix $M$ of the original states
\begin{align}
    M_{\rm LVP}^{(n=2)}&= \tanh{iG}\\
    &=\frac{2M}{1+M^2},
\end{align}
and for any copies $n$ we obtain
\begin{align}
    M_{\rm LVP}^{(n)}&= \tanh{\frac{inG}{2}}\\
    &=\frac{(1+M)^n-(1-M)^n}{(1+M)^n+(1-M)^n}.
\end{align}

\subsection{An example: one-dimensional transverse field Ising model}
The Hamiltonian of one-dimensional transverse field Ising model in Eq.~(\ref{eq:transverseising}) 
can be mapped into a free-fermionic Hamiltonian and the ground state (the Gibbs states) is characterized by a fermionic Gaussian pure state (the fermionic Gaussian states).
The correlation matrix of these states
is written as
\begin{align}
M_{\rm TFI}=
    \begin{bmatrix}
       \Pi_0 &\Pi_1 &\Pi_2&\cdots&\Pi_{N_A+N_B-1}\\
       \Pi_{-1} &\Pi_0 &\Pi_1&\cdots&\vdots\\
       \Pi_{-2} &\Pi_{-1} &\Pi_0&\cdots&\vdots\\
       &&\ddots&&\\
       \Pi_{1-N_A-N_B}&&\cdots&&\Pi_{0}
    \end{bmatrix},
\end{align}
where $\Pi_l$ is a $2\times 2$ matrix defined by
\begin{align}
    \Pi_l =
    \begin{bmatrix}
        0&g_l\\
        -g_{-l}&0
    \end{bmatrix}.
\end{align}
Assuming $N\rightarrow \infty$, $g_l$ can be written as an integral:
\begin{align}
    g_l&=\frac{i}{2\pi}\int_{-\pi}^\pi dq e^{iql}e^{i\theta_{q}}\frac{1-e^{-2 \beta \varepsilon_{q}}}{1+e^{-2 \beta \varepsilon_{q}}},
    \end{align}
    where $\varepsilon_q$ and $\theta_{q}$ are defined by
    \begin{align}
    \varepsilon_q&= \sqrt{(\cos{q}+\lambda)^2+(\sin{q})^2}\label{appeq:singleenergyTFIM}\\
    \cos{\theta_{q}}&=\frac{\cos{q}+\lambda}{\sqrt{(\cos{q}+\lambda)^2+(\sin{q})^2}}\\
    \sin{\theta_{q}}&=\frac{\sin{q}}{\sqrt{(\cos{q}+\lambda)^2+(\sin{q})^2}}.
\end{align}
If we consider the ground state $(\beta \rightarrow \infty )$ on a critical point $(\lambda=1)$, then
we can easily calculate the integral:
\begin{align}
    g_l = \frac{-i}{\pi (l+1/2)}.
\end{align}

\begin{figure}[htbp]
    \centering
  \includegraphics[width=0.6\textwidth]{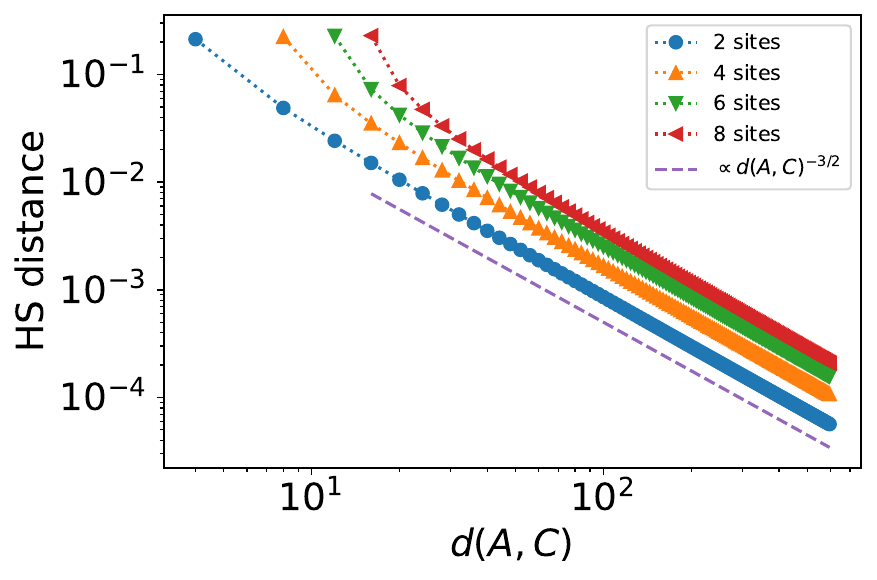}
  \caption{Hilbert-Schmidt distance between FVP and LVP against $d(A,C)$ for a noiseless case on a critical point. We fix $N_A=2,4,6,8$ and $N_C\rightarrow \infty$.
  }
  \label{Fig:HSdistanceinfty}
\end{figure}
It is known that the one-dimensional transverse field Ising model implies a power-law decay of the correlation function on a critical point ($\lambda=1$).
Figure~\ref{Fig:HSdistanceinfty} shows the Hilbert-Schmidt distance between FVP and LVP against $d(A,C)$  on a critical point. This graph actually shows the power-law decay and we can see that its power is equal to $3/2$. In this case, we take the limit of $N_C\rightarrow \infty$ and we can neglect the finite size effect of the total system.

\section{Error mitigation by LVP under simple noise model}
\subsection{Global depolarizing noise channel}
In this subsection, we derive the inequality of Theorem \ref{Theorem2} for the noisy ground states under the global depolarizing noise channel.
We consider the global depolarizing noise as a simple noise model:
\begin{align}
    \rho=(1-p)\rho_0+p \frac{\mathbb{I}}{2^N},
\end{align}
where $p$ ($0\le p\le 1$) is an error rate, $\rho_0$ denotes a noiseless pure state, and $\mathbb{I}$ is  the identity operator on the total system.
This noise model is widely seen as an effective noise model for a quantum circuit with deep depth~\cite{PhysRevLett.127.270502,PhysRevE.104.035309,tsubouchi2023universal,Qin_2023,foldager2023shallow}.

In the following analysis, we assume the two assumptions: (i) $N\gg1$ and (ii) $\tr[(\rho^{(A+B)}_0)^2]\gg \frac{1}{2^{N_A+N_B}}$ and $\tr_{A+B}[(\rho^{(A+B)}_0)^2]$ is independent of $N_A$ and $N_B$.
The first assumption always holds when we choose a sufficiently large $N$, and we can neglect the terms with an exponentially small factor $1/2^N$.
The second assumption means that $\rho^{(A+B)}_0$, a reduced density operator of $\rho_0$ in regions $A$ and $B$, has a relatively large value of the purity compared with that of the completely mixed state (note that $\frac{1}{2^{N_A+N_B}}$ is the value of the purity for the completely mixed states). For a unique ground state of a gapped system, which satisfies Eq.~(\ref{Eq:exponentialclustering}) in the main text, this assumption holds because the purity can be bounded by a constant independent of $N_A+N_B$ due to the the area law of the entanglement entropy.

We first consider 
the expectation values computed by the FVP as
\begin{align}
    \frac{\tr[\rho ^2 o_A]}{\mathrm{Tr}[\rho ^2]}&=\frac{(1-p)\left[1-\frac{(2^N-2)p}{2^N}\right]}{(1-p)\left[1-\frac{(2^N-2)p}{2^N}\right]+\frac{p^2}{2^N}}\tr[\rho_0  o_A]+\frac{1}{2^{2N}}\frac{p^2\tr[  o_A]}{(1-p)\left[1-\frac{(2^N-2)p}{2^N}\right]+\frac{p^2}{2^N}}\\
    &=\tr[\rho_0  o_A]-\frac{\frac{p^2}{2^N}\tr[\rho_0  o_A]}{(1-p)\left[1-\frac{(2^N-2)p}{2^N}\right]+\frac{p^2}{2^N}}+\frac{1}{2^{2N}}\frac{p^2\tr[  o_A]}{(1-p)\left[1-\frac{(2^N-2)p}{2^N}\right]+\frac{p^2}{2^N}}\\
    &=\tr[\rho_0  o_A] + \delta_1,
\end{align}
where $o_A$ is an arbitrary observable
supported on the region $A$.
$\delta_1$ is exponentially small regarding $N$ under the first assumption.

On the other hand, we calculate the expectation value of $o_A$ via LVP from a noisy reduced density operator
    $\rho^{(A+B)}=(1-p)\rho^{(A+B)}_0+p \frac{\mathbb{I}^{(A+B)}}{2^{N_A+N_B}}$, where we define $\mathbb{I}^{(A+B)}$ by the identity operator on the region $A$ and $B$, as
\begin{align}
    \frac{\tr_{A+B}[(\rho^{(A+B)})^2o_A]}{\mathrm{Tr}_{A+B}[(\rho^{(A+B)}) ^2]}
    &=\frac{(1-p)^2\tr_{A+B}[(\rho^{(A+B)}_0)^2o_A]  + 2p(1-p) \tr_{A+B}[\rho^{(A+B)}_0o_A] 2^{-(N_A+N_B)}+p^2\tr_{A+B}[o_A]2^{-2(N_A+N_B)} }{(1-p)^2\mathrm{Tr}_{A+B}[(\rho^{(A+B)}_0) ^2]+p(2-p) 2^{-(N_A+N_B)} }\\
    &=\frac{\tr_{A+B}[(\rho^{(A+B)}_0)^2o_A]}{\mathrm{Tr}_{A+B}[(\rho^{(A+B)}_0) ^2]}+ \frac{ p(1-p)}{2^{N_A+N_B}(1-p)^2\mathrm{Tr}_{A+B}[(\rho^{(A+B)}_0) ^2]+p(2-p) } \notag\\
    &\qquad\times\left[ 2 \tr_{A+B}[\rho^{(A+B)}_0o_A] - \frac{\tr_{A+B}[(\rho^{(A+B)}_0)^2o_A]}{\mathrm{Tr}_{A+B}[(\rho^{(A+B)}_0) ^2]}\frac{2-p }{1-p}\right]\notag\\
    &\qquad+\frac{p^2\tr_{A+B}[o_A]2^{-2(N_A+N_B)} }{(1-p)^2\mathrm{Tr}_{A+B}[(\rho^{(A+B)}_0) ^2]+p(2-p) 2^{-(N_A+N_B)} }\\
    &=\frac{\tr_{A+B}[(\rho^{(A+B)}_0)^2o_A]}{\mathrm{Tr}_{A+B}[(\rho^{(A+B)}_0) ^2]}+\delta_2.
\end{align}
Here, $\delta_2$ is exponentially small regarding $N_A+N_B$ from the second assumption:
\begin{align}
    (0\le)\frac{ p(1-p)}{2^{N_A+N_B}(1-p)^2\mathrm{Tr}_{A+B}[(\rho^{(A+B)}_0) ^2]+p(2-p) }\le \frac{p}{1-p} \frac{1}{2^{N_A+N_B}\mathrm{Tr}_{A+B}[(\rho^{(A+B)}_0) ^2]}=O(e^{-(N_A+N_B)}).
\end{align}
%
Therefore we obtain
\begin{align}
    |D^{(n=2)}(o_A)|&=\left|\frac{\mathrm{Tr}[\rho ^2o_A]}{\mathrm{Tr}[\rho ^2]}-\frac{\mathrm{Tr}_{A+B}[(\rho^{(A+B)}) ^2o_A]}{\mathrm{Tr}_{A+B}[(\rho^{(A+B)}) ^2]} \right|\\
    &\le |D^{(n=2)}_0(o_A)|+ \left|\delta_1\right| + \left|\delta_2\right|.
    \label{Seq:resultGlobaldepolarizing}
\end{align}
The first term of Eq.~(\ref{Seq:resultGlobaldepolarizing}) $|D^{(n=2)}_0(o_A)|$ exponentially decays in terms of $d(A,C)$ from Eq.~(\ref{Eq:exponentialclustering}) in the main text, and the rest of  the terms in Eq.~(\ref{Seq:resultGlobaldepolarizing}) can be negligible ($\delta_1=O(e^{-N})$ and $\delta_2 =O(e^{-(N_A+N_B)})$) because of the first and second assumption.

\subsubsection{Comparison of measurement costs between LVP and FVP under global depolarizing noise}
The measurement costs of LVP are typically larger than those of FVP:
\begin{align}
    \mathrm{Tr}[\rho ^2]&\simeq (1-p)^2\\
     \mathrm{Tr}[(\rho^{(A+B)}) ^2]&\simeq (1-p)^2\mathrm{Tr}[(\rho^{(A+B)}_0) ^2]\le (1-p)^2.
\end{align}
Nonetheless the difference is not crucial compared with the effect of errors in the controlled-derangement circuits, which gives exponentially large measurement costs for FVP~\cite{vikstål2022study} as described in the main text, if the above second assumption that $\mathrm{Tr}[(\rho^{(A+B)}_0) ^2]$ is independent of $N_A+N_B$ holds.

\subsection{Local and independent noise channel}
In this subsection, we investigate the effect of noise on the correlation function in Eq.~(\ref{eq:noisytwopointcorrelationfunction}) under local and independent noise channels such as single-qubit depolarizing noise.
We consider a gate-wise error model~\cite{PhysRevX.8.031027}:
\begin{align}
\mathcal{E}_i= (1-p)\mathcal{I}+ p \mathcal{E}'_i,
\end{align}
where $p$ is an error rate and $\mathcal{I}$ and $\mathcal{E}'_i$ are the identity map and an error channel, respectively.
The total error map is defined by
\begin{align}
\mathcal{E}&=\prod_{i=1}^{N_g}\mathcal{E}_i=\prod_{i=1}^{N_g} [(1-p)\mathcal{I}_i+ p \mathcal{E}'_i]\\
&=\sum_{i=0}^{N_g}\binom{N_g}{i}(1-p)^{N_g-i}p^{i}\mathcal{X}_i\\
&=(1-N_gp)\mathcal{I}+ N_gp \mathcal{X}_1 +O(p^2),
\end{align}
where $\mathcal{X}_i$ is the average error channel of $i$-th times.
Here, we assume that an error rate $p$ is sufficiently small and we are interested in the first-order term of $p$.
Under this assumption, we obtain a noisy state $\rho$ from a noiseless state $\rho_0$
\begin{align}
    \rho=\mathcal{E}[\rho_0]&=(1-N_gp)\rho_0+N_gp \mathcal{X}_1[\rho_0]+O(p^2)\\
    &=(1-N_gp)\rho_0+N_gp \rho_1+O(p^2),
    \label{Eq:gatewiseerror}
\end{align}
where we define $\rho_1=\mathcal{X}_1[\rho_0]$.
From this representation, we can calculate $\rho ^n$ as
\begin{align}
    \rho ^n&=\left[(1-N_gp)\rho_0+N_gp \rho_1\right]^n\\
    &=(1-nN_g p)\rho_0+N_gp(\rho_0\rho_1+(n-2)\rho_0\rho_1\rho_0+\rho_1\rho_0) +O(p^2) \\
    &=(1-nN_g p)\rho_0+nN_gp\rho_1' +O(p^2),
\end{align}
and $\mathrm{Tr}[\rho ^n]$ as
\begin{align}
    \mathrm{Tr}[\rho ^n]&=(1-nN_g p)+nN_gp\mathrm{Tr}[\rho_1'] +O(p^2),
\end{align}
where we define $\rho_1'$ by $\rho_1'=\frac{1}{n}[\rho_0\rho_1+(n-2)\rho_0\rho_1\rho_0+\rho_1\rho_0$].
Using these equations, we obtain the purified density matrix of FVP $\rho_{\rm FVP}$:
\begin{align}
    \rho_{\rm FVP}=\frac{\rho ^n}{\mathrm{Tr}[\rho ^n]}&=\frac{(1-nN_g p)\rho_0+nN_gp\rho_1'}{(1-nN_g p)+nN_gp\mathrm{Tr}[\rho_1']}+O(p^2)\\
    &=(1-nN_gp\mathrm{Tr}[\rho_1'])\rho_0+nN_gp\rho_1'+O(p^2).
\end{align}
On the other hand, the reduced density matrix in region $A$ and $B$ is given by
\begin{align}
    \rho^{(A+B)}=(1-N_gp)\rho^{(A+B)}_0+N_gp\rho^{(A+B)}_1 +O(p^2),
\end{align}
where we define $\rho^{(A+B)}_0=\tr_C[\rho_0]$ and $\rho^{(A+B)}_1=\tr_C[\rho_1]$ for simplicity.
From the following equations
\begin{align}
    (\rho^{(A+B)})^n&=(1-nN_gp)(\rho^{(A+B)}_0)^n+ N_gp(\sum_{k=0}^{n-1}(\rho^{(A+B)}_0)^k\rho^{(A+B)}_1(\rho^{(A+B)}_0)^{n-k-1}) +O(p^2),
    \end{align}
    and
    \begin{align}
    \mathrm{Tr}_{A+B}[(\rho^{(A+B)}) ^n]&=(1-nN_gp)\mathrm{Tr}_{A+B}[(\rho^{(A+B)}_0) ^n]+nN_gp\mathrm{Tr}_{A+B}[(\rho^{(A+B)}_0) ^{n-1}\rho^{(A+B)}_1 ]+O(p^2),
\end{align}
the purified density matrix of LVP can be calculated as
\begin{align}
    \frac{(\rho^{(A+B)})^n}{\mathrm{Tr}_{A+B}[(\rho^{(A+B)}) ^n]}
    =&\frac{(1-nN_gp)(\rho^{(A+B)}_0)^n+ N_gp(\sum_{k=0}^{n-1}(\rho^{(A+B)}_0)^k\rho^{(A+B)}_1(\rho^{(A+B)}_0)^{n-k-1})}{\mathrm{Tr}_{A+B}[(\rho^{(A+B)}_0) ^n]}\notag\\
    &+nN_gp\frac{1-\frac{\mathrm{Tr}_{A+B}[(\rho^{(A+B)}_0) ^{n-1}\rho^{(A+B)}_1 ]}{\mathrm{Tr}_{A+B}[(\rho^{(A+B)}_0) ^n]}}{{\mathrm{Tr}_{A+B}[(\rho^{(A+B)}_0) ^n]}}(\rho^{(A+B)}_0)^n +O(p^2)\\
    =&\frac{(\rho^{(A+B)}_0)^n+ N_gp(\sum_{k=0}^{n-1}(\rho^{(A+B)}_0)^k\rho^{(A+B)}_1(\rho^{(A+B)}_0)^{n-k-1})}{\mathrm{Tr}_{A+B}[(\rho^{(A+B)}_0) ^n]}\notag\\
    &-nN_gp\frac{\mathrm{Tr}_{A+B}[(\rho^{(A+B)}_0) ^{n-1}\rho^{(A+B)}_1 ]}{\mathrm{Tr}_{A+B}[(\rho^{(A+B)}_0) ^n]^2}(\rho^{(A+B)}_0)^n +O(p^2).
\end{align}
Therefore the difference of the noisy state in Eq.~(\ref{Eq:gatewiseerror}) can be bounded as
\begin{align}
    |D^{(n)}(o_A)|&=\left|\frac{\mathrm{Tr}[\rho ^no_A]}{\mathrm{Tr}[\rho ^n]}-\frac{\mathrm{Tr}_{A+B}[(\rho^{(A+B)}) ^no_A]}{\mathrm{Tr}_{A+B}[(\rho^{(A+B)}) ^n]} \right|\\
    &\le |D_0^{(n)}(o_A)|+nN_gp\left|-\tr[\rho_1']\tr[\rho_0 o_A]+\mathrm{Tr}[\rho_1'o_A]+\frac{\mathrm{Tr}_{A+B}[(\rho^{(A+B)}_0) ^{n-1}\rho^{(A+B)}_1 ]}{\mathrm{Tr}_{A+B}[(\rho^{(A+B)}_0) ^n]^2}\mathrm{Tr}_{A+B}[(\rho^{(A+B)}_0)^no_A]\right.\notag\\
    &\left.-\frac{1}{n}\frac{ \mathrm{Tr}_{A+B}[(\sum_{k=0}^{n-1}(\rho^{(A+B)}_0)^k\rho^{(A+B)}_1(\rho^{(A+B)}_0)^{n-k-1})o_A]}{\mathrm{Tr}_{A+B}[(\rho^{(A+B)}_0) ^n]} \right|+O(p^2).
    \label{app:localnoise}
\end{align}
The first term in Eq.~(\ref{app:localnoise}) is the noiseless correlation function $|D^{(n)}_0(o_A)|$, which decays exponentially in terms of $d(A,C)$ as shown in Theorem 2, and the second term in Eq.~(\ref{app:localnoise}) denotes the difference of the bias term between LVP and FVP due to an effect of noise up to a first-order approximation in $p$.
Although each bias term, such as $-\tr[\rho_1']\tr[\rho_0 o_A]+\mathrm{Tr}[\rho_1'o_A]$ for FVP, is not equal to zero in general, we numerically confirm that the difference of them converges to zero with $d(A,C)$ increasing, even if $p$ is not a small value.
We have explained this behavior in Fig.~\ref{Fig:expectationvalue} for the case of local depolarizing noise.
Additionally, we show numerical results for another noise model.
\begin{figure}[htbp]
    \centering
  \includegraphics[width=0.9\textwidth]{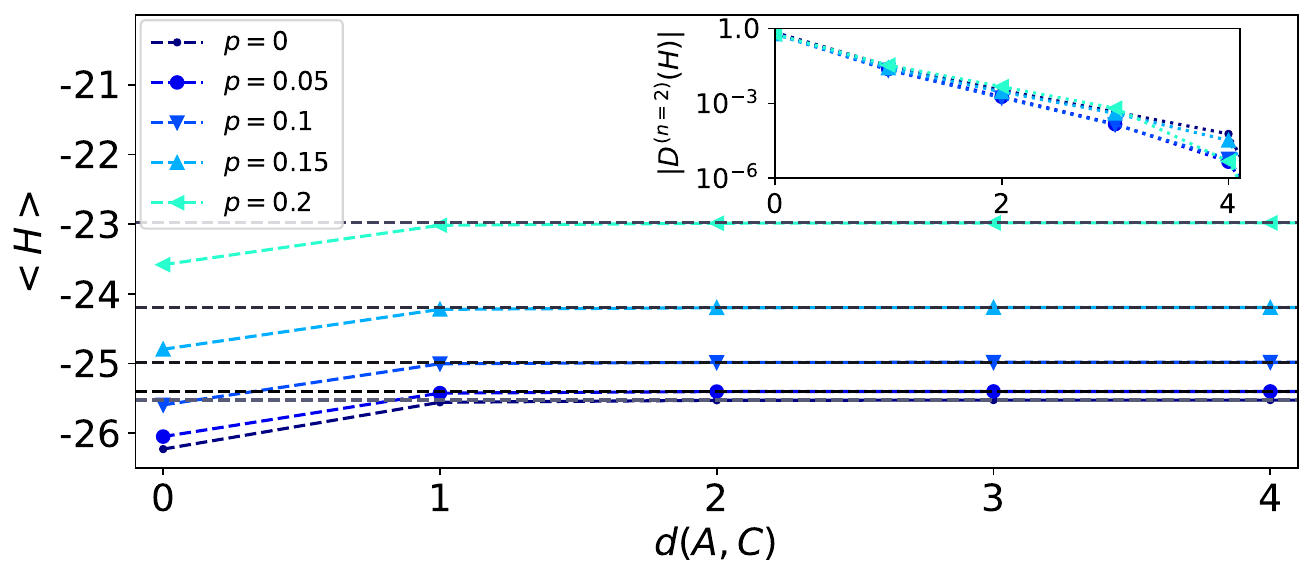}
  \caption{Expectation values of $H_{\rm TFI}$ estimated by the LVP against $d(A,C)$ for the ground state of one-dimensional transverse field Ising model on a non-critical point with local dephasing noise. We fix $n=2$, $N=12$, and $p=0,0.05,0.1,0.15,0.2$. 
  Horizontal dashed lines indicate the expectation values of FVP. (inset) The difference of the expectation values between FVP and LVP against $d(A,C)$. 
  }
  \label{Fig:mitigation_dephasing}
\end{figure}
Figure \ref{Fig:mitigation_dephasing} shows the results of similar calculation as in Fig.~\ref{Fig:expectationvalue}, with the local depolarizing noise replaced with single-qubit dephasing noise.
This graph also shows that the difference in the expectation values decays exponentially in $d(A,C)$.
Hence, under various local noise models, our LVP protocol has the ability to mitigate an effect of error as much as the original version of FVP, with exponentially small bias regarding $d(A,C)$.

\end{document}